\documentclass[acmsmall,authorversion]{acmart}

\usepackage{prelude}

\AtBeginDocument{%
  }

\setcopyright{rightsretained}
\acmDOI{10.1145/3632865}
\acmYear{2024}
\copyrightyear{2024}
\acmSubmissionID{popl24main-p107-p}
\acmJournal{PACMPL}
\acmVolume{8}
\acmNumber{POPL}
\acmArticle{23}
\acmMonth{1}
\received{2023-07-11}
\received[accepted]{2023-11-07}

\begin{document}

\title{A Core Calculus for Documents}
\subtitle{Or, Lambda: the Ultimate Document}

\author{Will Crichton}
\orcid{0000-0001-8639-6541}
\author{Shriram Krishnamurthi}
\orcid{0000-0001-5184-1975}
\affiliation{
  \department{Department of Computer Science}
  \institution{Brown University}           
  \city{Providence}
  \state{Rhode Island}
  \postcode{02912}
  \country{USA}                    
}
\email{wcrichto@brown.edu}

\begin{abstract}
Passive documents and active programs now widely comingle. Document languages include Turing-complete programming elements, and programming languages include sophisticated document notations. However, there are no formal foundations that model these languages. This matters because the interaction between document and program can be subtle and error-prone. In this paper we describe several such problems, then taxonomize and formalize document languages as levels of a document calculus. We employ the calculus as a foundation for implementing complex features such as reactivity, as well as for proving theorems about the boundary of content and computation. We intend for the document calculus to provide a theoretical basis for new document languages, and to assist designers in cleaning up the unsavory corners of existing languages. 
\end{abstract}

\ccsdesc[500]{Applied computing~Document scripting languages}
\ccsdesc[500]{Theory of computation~Type theory}

\keywords{document languages, markup, templates}


\maketitle

\section{Introduction}

We live in a golden age of document languages. We have time-tested methods for authoring stylized content, such as the widely-used Markdown language. Moreover, new commercially-backed languages like Typst\,\cite{maedje2022typst}, Markdoc\,\citeurl{https://markdoc.dev/}, Quarto\,\citeurl{https://quarto.org/}, and MDX\,\citeurl{https://mdxjs.com/} are providing ever more powerful ways of authoring documents. These languages are built within a rich tradition established by venerable systems like {\TeX}\,\cite{knuth1986texbook} and Scribe\,\cite{reid1980scribe}, and continuing with recent languages like Scribble\,\cite{flatt2009scribble} and Pandoc\,\citeurl{https://pandoc.org/}.

Many of these systems are as much \emph{programming} languages as document languages. Document authors want to systematically format data collections, abstract over similar text, create abbreviations, hide text for anonymous review, and so on. These tasks benefit from programmatic control over the document. 
Conversely, the lack of programmability in languages like HTML has generated a cottage industry of document \emph{meta}languages, often called \emph{template} languages. For example, general-purpose languages with templates include PHP, Javascript (with JSX), and Lisp (via quasiquotes). Specialized template languages include Jinja 
for Python and Liquid 
for Ruby. 

This proliferation of document languages raises foundational questions. What are the common characteristics of document languages? How do they relate? Are existing languages well-designed, or can we identify what appear to be flaws in their design? Given that documents are programs, can we reason about them? Our goal in this work is to shed light on these questions through the design of a \emph{core calculus for documents} --- that is, a formal model for the essential \emph{computational} features of a document language. This paper describes the document calculus in four parts:



\begin{enumerate}
    \item \textbf{We motivate the work with case studies about issues in the semantics of existing document languages} (\Cref{sec:the-bad-parts}).
    We show how document languages from both academia and industry can lead to unexpected behavior when composing content and computation. The case studies demonstrate the need to carefully study features like templates and interpolation.

    \item \textbf{We incrementally describe the formal semantics of the document calculus} (\Cref{sec:doc-calc}).
    We construct the semantics in eight levels of the document language design space drawn from two key dimensions: document domain (strings or trees) and document constructors (literals, programs, template literals, template programs). We show how our choice of semantics corresponds to real-world document languages and also mitigates the issues in \Cref{sec:the-bad-parts}.

    \item \textbf{We demonstrate how the document calculus can provide a foundation for modeling complex document features} (\Cref{sec:article-extensions}).
    We extend the calculus with three features: references, reforestation, and reactivity. Each feature is common to many document languages and stresses a different computational aspect of the calculus.

    \item \textbf{We use the document calculus to formally reason about document programs} (\Cref{sec:reasoning}).
    We formalize two useful theorems involving the document calculus. First, we show that our choice of template semantics produces well-typed terms. Second, we prove the correctness of a strategy for efficiently composing references and reactivity.
\end{enumerate}

\noindent We conclude with related work (\Cref{sec:relatedwork}) and implications for language design (\Cref{sec:discussion}).

\section{Document Languages: The Bad Parts}
\label{sec:the-bad-parts}

The foundational concept of all document languages is the \emph{template}. Templates are a kind of generalized data literal. Templates interleave computation (like expressions and variable bindings) into content (like strings and trees). The good part of a template is its brevity --- with appropriate concrete syntax, a template can be more concise than a equivalent program without templates. This section is about the bad parts: \emph{when templates go wrong.}

We present a series of case studies about how particular designs for template semantics cause problems at the boundary of content and computation. In each case study, we present a reasonable-looking program that works as expected. Then we show how a small change can quickly produce unreasonable results. These issues are not fatal flaws in the language; each has a workaround that would likely be known to seasoned users. Rather, we are simply drawing attention to common points of friction that could be both clarified and improved with a formal semantics.

\vspace{-0.5em}
\subsection{PHP and the Global Mutable Buffer}
\label{sec:php-case-study}

PHP\,\citeurl{https://www.php.net/} is a popular programming language for web servers. A PHP program is itself a document, as any text placed outside a \verb|<? tag ?>| is emitted to the client. However, templates in PHP are not pure. They do not construct values, but rather write to a global output buffer. This impure semantics requires that functions must be called in exactly the right place.

\vspace{-0.3em}
\begin{twocol}
    \begin{col}
        \begin{minipage}{0.9\linewidth}
    \begin{minted}{php}
<? function mkelems($list) { 
  foreach ($list as $x) { ?>
    <li><?= $x ?></li>
  <? }}
function mklist($list) { ?>
  <ul><? mkelems($list) ?></ul>
<? }
mklist(["Hello", "World"]) ?>
\end{minted}
        \end{minipage}
    \end{col}
    \begin{col}
        For example, consider a PHP program that factors a bulleted-list generator into two functions. On the left, the function \verb|mklist| wraps the result of \verb|mkelems| in a \verb|ul| tag. The \verb|mkelems| function loops through the list and generates an \verb|li| for each element. Text within a function but outside the question-mark-delimited ranges is emitted to the global buffer when the function is called.
    \end{col}
\end{twocol}
\vspace{-0.5em}

\begin{twocol}
    \begin{col}
        Now, say the programmer wants to extend \verb|mkelems| to return data about \verb|$list| that is added as an attribute to the \verb|ul| inside \verb|mklist|. The programmer can easily modify \verb|mkelems| to return e.g. the list length. But how should the programmer modify \verb|mklist|? They must call \verb|mkelems| \emph{before} generating the \verb|ul| in order to access the return value. But they also must call \verb|mkelems| \emph{after} generating the opening tag in order to position the list elements correctly. 
    \end{col}
    \begin{col}
        \flushright
        \begin{minipage}{0.9\linewidth}
        \begin{minted}{php}
<?  function mkelems($list) { 
  foreach ($list as $x) { ?>
    <li><?= $x ?></li>
  <? }
  return count($list);
}
function mklist($list) { 
  $n = mkelems($list); ?>
  <ul data-n="<?= $n ?>"> ?? </ul>
<? }
mklist(["Hello", "World"]) ?>
\end{minted}
    \end{minipage}\end{col}
\end{twocol}
\vspace{-0.5em}

The programmer cannot easily compose the template mechanism with other concerns like returning auxiliary information. It is possible to work around this limitation through careful use of PHP's output buffering facilities, but those APIs are also stateful and hence prone to errors. The key takeaway: an impure semantics for templates can restrict the ability of authors to design easily-composable document abstractions.

\subsection{React and the Unresponsive Component}
\label{sec:react-toc}

React\,\citeurl{https://reactjs.com/} is a popular Javascript framework for writing reactive interfaces in the browser. A React program is a tree of components that encapsulate the state and view for a visual object. Ideally, React re-renders a component when its dependencies update. However, the kinds of components used in documents cannot always express their dependencies in terms understood by the framework.

\vspace{-0.3em}
\begin{twocol}
\begin{col}
\begin{minipage}{0.95\linewidth}
\begin{minted}{jsxlexer.py:JsxLexer -x}
function Toc() {
  let [hdrs, setHdrs] = useState(()=>[]);
  useEffect(() => {
    let h1s = Array.from(
      document.querySelectorAll('h1'));
    setHdrs(h1s.map(node => 
      node.textContent));
  }, []);
  return <ul>
    {hdrs.map(text => <li>{text}</li>)}
  </ul>;
}
\end{minted}
\end{minipage}
\end{col}   
\begin{col}
    Say a programmer wants to implement a table of contents. A na\"ive implementation would be the program on the left. The \verb|headings| array contains a list of the headings in the document, initially empty. The function \verb|useEffect| indicates that the provided callback should execute after the component renders. The callback uses the browser's DOM API to find all headers in the document. It extracts their text content and saves that array as local state. The returned template creates a bulleted list with a bullet per heading.
\end{col}
\end{twocol}

\vspace{-0.7em}
\begin{twocol}
\begin{col}
Here is an example application that uses the component. The \verb|App| component creates a local boolean state \verb|show| that is initialized to false. Then the template returns both a persistent header and a conditional header. A button is rendered that changes the condition on click, and then the table of contents is rendered. 
\end{col}
\begin{col}
\flushright
\begin{minipage}{0.95\textwidth}
\begin{minted}{jsxlexer.py:JsxLexer -x}
function App() {
  let [show, setShow] = useState(false);
  return <><Toc />
    <h1>Introduction</h1>
    {show ? <h1>Appendix</h1> : null}
    <button onClick={()=>setShow(!show)}>
      {show ? "Hide" : "Show"} Appendix
    </button></>;
}
\end{minted}    
\end{minipage}
\end{col}
\end{twocol}

\vspace{-0.7em}
\begin{twocol}
    \begin{col}
    \begin{minipage}{0.95\textwidth}
        \begin{tcolorbox}[spartan,
  frame empty,
  boxsep=0mm,
  left=1mm,right=1mm,top=1mm,bottom=1mm,
  colback=bg]
\centering
\includegraphics[width=0.9\textwidth]{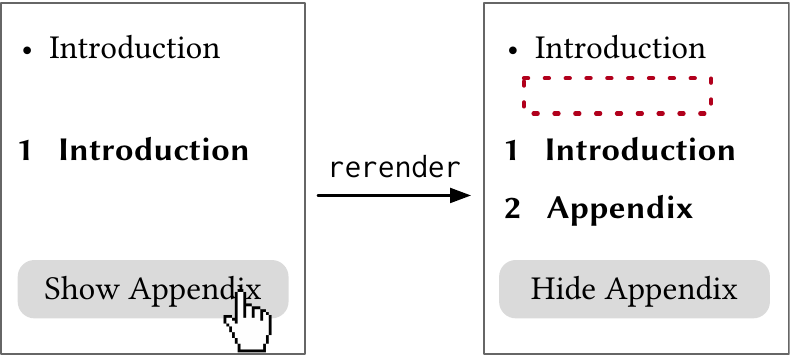}
\end{tcolorbox}
\end{minipage}
    \end{col}
    \begin{col}
        This application will correctly render on the first pass, showing a table of contents with one bullet for \quot{Introduction}. However, when the user clicks on the toggle button, the appendix header will appear, but the table of contents will not update, indicated by the dotted red rectangle.
    \end{col}
\end{twocol}

The issue is that the header-query computation in the \verb|Toc| component is not legible to the reactive runtime --- React simply cannot express the concept that a component is dependent on the content of other components' views. Therefore, React does not know to re-render the table of contents when \verb|App| changes its heading structure. And this issue is not a chance mistake --- as of September 2023, the top five results on Google for ``react table of contents'' are tutorials that all recommend this strategy. ToC implementations in other reactive frameworks like Svelte work similarly\,\citeurlwith{https://github.com/janosh/svelte-toc/blob/cc7f3d86b047c66b22547dabf9c0b34f32fa158c/src/lib/Toc.svelte\#L50}{janosh/svelte-toc}.
The key takeaway: computations over documents can have subtle dependency structures, which easily leads to bugs when combined with reactivity.

\subsection{Scribble and the Improper Loop}
\label{sec:scribble-study}

Scribble\,\cite{flatt2009scribble} is a document-oriented dialect of Racket. Scribble provides a template language that can interpolate computation via @-expressions, which desugar into standard Racket code. This desugaring can produce unexpected interactions with macros.
For example, this program uses the \verb|for/list| macro to map over a list of pairs:

\begin{twocol}
\begin{col}
\begin{minipage}{0.95\linewidth}
\begin{minted}{racket}
@(define pairs 
  (list (list "A" "B") (list "C" "D")))
@itemlist{
  @for/list[([p pairs])]{
    @item{@(car p) @(cadr p)}
  }
}
\end{minted}
\end{minipage}
\end{col}
\begin{col}
In this example, the \verb|@itemlist| represents the list container, and \verb|@item| represents a bullet in the list. The for-loop produces one bullet for each pair, creating the list:
\begin{itemize}
\item A B
\item C D
\end{itemize}
\end{col}
\end{twocol}

\begin{twocol}
\begin{col}
Now say that the programmer wanted to change the code to flatten the list. A programmer might expect that factoring the \verb|car| and \verb|cadr| into separate \verb|@item|s should accomplish this task. However, Scribble instead drops the first bullet from each iteration, producing this list:
\begin{itemize}
    \item B
    \item D
\end{itemize}
\end{col}
\begin{col}
\flushright
\begin{minipage}{0.95\linewidth}
\begin{minted}{racket}
@(define pairs 
  (list (list "A" "B") (list "C" "D")))
@itemlist{
  @for/list[([p pairs])]{
    @item{@(car p)}
    @item{@(cadr p)}
  }
}
\end{minted}   
\end{minipage}
\end{col}
\end{twocol}

\vspace{0.25em}
\noindent The cause of the bug is more apparent in this expression within the desugared Racket code:
\begin{minted}{racket}
(for/list ([(p pairs)])
  (item (car p)) '"\n" (item (cadr p)))
\end{minted}
The issue is that \verb|for/list| permits a ``body'' of s-expressions, where only the final s-expression becomes the value for each iteration. Scribble's @-expression desugaring directly ``pastes'' the sequence of template elements into the \verb|for/list| body, causing most of the template elements to be dropped. Notably, Racket's \verb|web-server| library uses Scribble's @-expressions, and its documentation cites this bug as a common ``gotcha'' for users\,\cite[\S 7.3.2]{racket-web-server}. The key takeaway: the desugaring of templates to terms requires careful scrutiny to understand how it composes with other language features.

\section{The Document Calculus}
\label{sec:doc-calc}

\Cref{sec:the-bad-parts} shows that the complexity of a document language lies in more than its syntax: its semantics affect how well parts of the language compose together. However, a challenge in conceptualizing document language semantics is that there exists no formal foundations for describing how a document language works. Our work aims to establish such a foundation by designing a \emph{document calculus}, or a formal semantics for the core computational aspects of document languages.

First, we will establish a scope by asking: what is a document, and what is a document language? Within this paper, we consider a document to be ``structured prose,'' that is, plain text optionally augmented with styles (e.g., bold or italics) or hierarchy (e.g., paragraphs or sections) and interspersed with figures (e.g., images or tables). This definition includes objects like academic papers and news articles, and it excludes objects like source code, spreadsheets, and computational notebooks. 
It is useful to restrict the scope of documents because (a) many languages are often used to generate objects in the former set, and (b) those languages have commonalities which have not yet been carefully scrutinized via the lens of PL formalism, unlike e.g. spreadsheets\,\cite{DBLP:journals/vlc/BockBSTT20}. We will give a formal definition of structured prose in the ensuing sections.

\newcommand{\usemark}[1]{\textit{#1}}
\begin{table}[t]
    \centering
\setcellgapes{2pt}
\makegapedcells
\begin{tabular}{c|p{1.3cm}|p{0.9cm}|p{3.5cm}|p{5cm}}
     \textbf{Domain} & \textbf{Ctors} & \textbf{Model} & \textbf{Example Languages} & \textbf{Example Syntax} \\\hline
     \multirow{4}{*}{\hyperref[sec:string-calc]{String}} & \hyperref[sec:string-lit]{Literal} & 
         $\dstrid$ & 
         {\small Text files, \usemark{quoted strings}} & 
         \mintinline{js}|"Hello World"| \\\cline{2-5}
     & \hyperref[sec:string-prog]{Program} & 
        $\dstrprog$ &
        {\small PLs with string APIs, such as \usemark{Javascript}} & 
        \mintinline{js}|"Hello" + " World"| \\ \cline{2-5}
     & \hyperref[sec:string-tmpl-lit]{Template Literal} & 
        $\dstrtmplstr$ &
        {\small C \verb|printf|, Python f-strings, \usemark{Javascript template literals}, Perl interpolated strings} & 
        \vspace{-0.5em}\begin{minipage}{5cm}
\begin{minted}[linenos=false,xleftmargin=0pt]{js}
let world = "World";
`Hello ${world}`
\end{minted}
\end{minipage}  \\\cline{2-5}
     & \hyperref[sec:string-tmpl-prog]{Template Program} & 
        $\dstrtmpllang$ &
        {\small C preprocessor, PHP, LaTeX, \usemark{Jinja} (Python), Liquid (Ruby), Handlebars (Js)} & 
        \vspace{-0.5em}\begin{minipage}{5cm}
\begin{minted}[linenos=false,xleftmargin=0pt]{jinja}
{% set world = "World" %}
Hello {{ world }}
\end{minted}
\end{minipage} \\ \hline
     \multirow{4}{*}{\hyperref[sec:article-calc]{Article}} & \hyperref[sec:article-lit]{Literal} & 
        $\dartid$ &
        {\small \usemark{CommonMark Markdown}, Pandoc Markdown, HTML, XML} & 
        \mintinline{md}|- Hello **World**| \\\cline{2-5}
     & \hyperref[sec:article-prog]{Program} & 
        $\dartprog$ &
        {\small PLs with document APIs, such as \usemark{Javascript}} & 
        \vspace{-0.5em}\begin{minipage}{5cm}
\begin{minted}[linenos=false,xleftmargin=0pt]{js}
var ul = 
  document.createElement("ul");
// ...
\end{minted}
\end{minipage}\vspace{0.3em} \\ \cline{2-5}
     & \hyperref[sec:article-tmpl-lit]{Template Literal} & 
        $\darttmpllit$ &
        {\small JSX Javascript, Scala 2, VB.NET, \usemark{Scribble Racket}, MDX Markdown, Lisp quasiquotes} &
        \vspace{-0.5em}\begin{minipage}{5cm}
\begin{minted}[linenos=false,xleftmargin=0pt]{racket}
@(define world "World")
@itemlist{@item{
  Hello @bold{@world}}}
\end{minted}
\end{minipage} \\ \cline{2-5}
    & \hyperref[sec:article-tmpl-prog]{Template Program} & 
        $\darttmplprog$ &
        {\small \usemark{Typst}, Razor C\#, Svelte Javascript, Markdoc Markdown} &
        \vspace{-0.5em}\begin{minipage}{5cm}
\begin{lstlisting}[language=typst,numbers=none,xleftmargin=0pt]
#let world = [World]
- Hello *#world*
\end{lstlisting}
\end{minipage} \\\hline
    \end{tabular}
    \caption{Taxonomy of \dl{}s based on the document calculus. A level is given as a particular pair document domain and document constructors. Each level corresponds to a specific family of existing \dl{}s, with hyperlinks to the corresponding section of the paper. Example real-world languages in the family are provided along with a sample syntax from one of those languages (indicated with italics).}
    \label{tab:taxonomy}
\end{table}

\newcommand{\kdomain}{\mathcal{M}}
\newcommand{\kconstructor}{\mathcal{C}}
\newcommand{\klevel}{\mathcal{L}}

A document language, then, is a programming language that is commonly used to generate documents. Some document languages are specially designed for documents, such as Markdown, while others are general-purpose but commonly used to generate documents, such as PHP. In reviewing the space of existing document languages, our key insight is that the design space can be decomposed along two dimensions:
\begin{itemize}
    \item \textbf{Document domain}: the type of document generated by a document language. There are two main document domains: plain strings, and annotated trees of strings (which we call ``articles''). Formally, we write this as:
    $$
    \msf{Domain}~\kdomain ::= \str \mid \karticle
    $$
    
    \item \textbf{Document constructors}: the expressiveness of the operations for constructing a value in the document domain. There are four categories of expressiveness: literals (no computation), programs (computation over literals), templates literals (literals with interpolation), and template programs (literals with loops and variables). Formally, we write this as:
    $$
    \msf{Constructor}~\kconstructor ::= \lit \mid \prog \mid \tlit \mid \tprog
    $$
\end{itemize}

A language level is defined as an element of the cross product $\msf{Domain} \times \msf{Constructor}$. This taxonomy is useful because each level corresponds to multiple widely-used document languages, as shown in \Cref{tab:taxonomy}. This taxonomy also provides a natural progression for the development of the document calculus: starting with string literals, we can add successively more features until reaching article template programs. 

The document calculus therefore consists of $2 \times 4 = 8$ levels. Each level of the document calculus is written as $\dlm{\kdomain}{\kconstructor}$, which consists of a document domain $\kdomain$, document expressions  $\expr^\kdomain_\kconstructor$ for the domain $\kdomain$ with constructors $\kconstructor$, and a semantics that relates the two. This section will present each level by first giving examples of real-world document languages at that level, and then providing a formal definition for the level. We also provide an OCaml implementation of these semantics in the supplementary materials, using open functions\,\cite{loh2006open} to match the incremental presentation of the semantics.

As with any model, the document calculus focuses on some aspects to the exclusion of others. For instance, concrete syntax is an essential aspect of any document language. However, our focus is on the computational aspects, so we postpone discussion of syntax to \Cref{sec:concrete-syntax}. As another example, we only model computation as interpolation, binding, conditionals, and loops. We believe these constructs capture the core commonalities amongst the languages in \Cref{tab:taxonomy}. But this decision inevitably omits various features we consider ancillary to the document language design space. For instance, template DSLs like Liquid and Handlebars provide a special facility for accessing the index of a loop iteration, but we do not include that feature in the document calculus.

\subsection{The String Calculus}
\label{sec:string-calc}

A string $s \in \str$ is a sequence of characters $c \in \chr$, such as ``x'' or ``$\Gamma$'' or ``\includegraphics[height=1em]{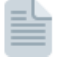}''. We will present a sequence of document calculi $\dlm{\str}{\bullet}$ in the string domain.

\subsubsection{String Literals} 
\label{sec:string-lit}

The first level is the string literal calculus $\dstrid$. For example, text files (left) and string literals (right) are both examples of document languages at this level:

\begin{twocol}
    \begin{col}
        \begin{minted}{text}
I'm suspicious of "strings".        
        \end{minted}
    \end{col}
    \begin{col}
\begin{minted}{python}
"I'm suspicious of \"strings\"."
\end{minted}
    \end{col}
\end{twocol}

\noindent Formally, $\dstrid$ has no computation and therefore the simplest semantics:
\begin{align*}
    \expr^\str_\lit ~ e &::= s \\
    \ty^\str_\lit ~ \tau &::= \str
\end{align*}

\subsubsection{String Programs}
\label{sec:string-prog}

The next level is the string program calculus $\dstrprog$, which supports both string-specific operations (like concatenation) and domain-general operations (like variable binding). This level models general-purpose programming languages with support for strings, such as this Javascript program (left) and Q program (right):

\begin{twocol}
\begin{col}
  \begin{minted}{js}
// Javascript 
let x = "a";
x + "b" + x
\end{minted}  
\end{col}
\begin{col}
\begin{minted}{apl}
/ Q
x: "a"
x, "b", x        
\end{minted}
\end{col}
\end{twocol}

\newcommand{\inject}[3]{\msfb{inject} ~ {#1} ~ \msfb{at} ~ {#2} ~ \msfb{as} ~ {#3}}
\newcommand{\tysum}[1]{\langle{#1}\rangle}
\newcommand{\typrod}[1]{\{{#1}\}}
\newcommand{\recsum}[3]{{#1}_{#2} ~ {#3}}

\noindent Formally, $\dstrprog$ is System F with a base type of strings and a few features relevant for later levels. Namely, string concatenation, fixpoints, sums, products, recursive types, and existential types:
\begin{align*}
    \mathclap{\hspace{25em} \msf{Variable} ~ x \hspace{1.5em} \msf{Type ~ Variable} ~ \alpha \hspace{1.5em}    \msf{Label} ~ \ell} \\
    \expr^\str_\prog ~ e ::=~ &e^\str_\lit \mid \concat{e_1}{e_2} \mid \lambda (x : \tau).~ e \mid e_1 ~ e_2 \mid x \mid \msfb{fix}(x : \tau).~ e \mid \letexp{x}{e_1}{e_2} \mid \\
    & \{(\ell: e_\ell)^*\} \mid e.\ell \mid \inject{e}{\ell}{\tau} \mid \msfb{case} ~ e ~ \{(\ell(x) \Rightarrow e_\ell)^*\} \\
    & \msfb{fold}_\tau ~ e \mid \msfb{unfold}_\tau ~ e \mid \Lambda \alpha. ~ e \mid e[\tau] \mid \msfb{pack} ~ e ~ \msfb{as} ~ \exists \alpha.\tau \mid \msfb{unpack} ~ (x, \alpha) = e_1 ~ \msfb{in} ~ e_2 \\
    \ty^\str_\prog \tau ::=~ &\tau^\str_\lit \mid \tau_1 \rightarrow \tau_2 \mid \{(\ell: \tau_\ell)^*\} \mid \tysum{(\ell: \tau_\ell)^*} \mid \forall \alpha. \tau \mid \mu \alpha. \tau  \mid \exists \alpha. \tau  \mid \alpha\\
\end{align*}

\noindent The static and dynamic semantics of all the features are standard, so we provide them in \appref{sec:additional-rules} for reference. But as a simple example, the ``aba'' program can be written as as follows:
$$\letexp{x}{\quot{a}}{\concat{\concat{x}{\quot{b}}}{x}} \bigstepto \quot{aba}$$ 

\noindent Note that not all of these features are essential for dealing with strings, e.g., recursive types will only be useful for representing tree documents. But rather than scattering this part of the language definition throughout the levels, we opted to introduce all the relevant System F features here. This enables the development of later levels to focus more on the purely document-related features.

In the remainder of the paper, we will refer to an assumed standard library of common types and operations containing the following:
\begin{align*}
  \tylist{\tau} &\triangleq \mu \alpha.~ \tysum{\msf{nil}: () \mid \msf{cons}: \typrod{\msf{hd}: \tau,~ \msf{tail}: \alpha}} \\
  \recsum{\ell}{\mu \alpha. \tau}{e} &\triangleq \msfb{fold}_{\mu \alpha.\tau} ~ \inject{e}{\ell}{\tau[\alpha \rightarrow \mu \alpha. \tau]} \\
  \msf{map} &: \forall \alpha, \beta. ~ (\alpha \rightarrow \beta) \rightarrow \alpha~\msf{list} \rightarrow \beta~\msf{list} \\
  \msf{flatten} &: \forall \alpha.~ \alpha~\msf{list}~\msf{list} \rightarrow \alpha~\msf{list} \\
  \msf{append} &: \forall \alpha.~ \alpha~\msf{list} \rightarrow \alpha~\msf{list} \rightarrow \alpha~\msf{list} \\
  \msf{join} &: \tylist{\str} \rightarrow \str
\end{align*}

\subsubsection{String Template Literals}
\label{sec:string-tmpl-lit}

The next level is the string template literal calculus $\dstrtmplstr$. For example, the ``aba'' program can be written in Javascript (left) and Python (right) programs using string template literals:

\begin{twocol}
\begin{col}
    \begin{minted}{js}
// Javascript
let x = "a";
`${x}b${x}`
    \end{minted}
\end{col}
\begin{col}
    \begin{minted}{python}
# Python
x = "a"
f"{x}b{x}"
    \end{minted}
\end{col}
\end{twocol}

\noindent
String template literals are variously called ``template strings'', ``format strings'', and ``interpolated strings.'' We specifically use the term ``string template \emph{literals}'' to draw a distinction with ``string template \emph{programs}''. Template literals only support positional interpolation of expressions, while template programs support additional template-level features such as variable-bindings and loops. This distinction is useful because both template literals and template programs can be found in real-world systems. 

Formally, templates in $\dstrtmplstr$ are a list of template parts which can be either literals (strings) or interpolated expressions. Within an expression, a template is invoked with the $\msfb{strtpl}$ operator:
\begin{align*}
    \kttext^\str_\tlit ~ t &::= [p^*] \\
    \ktinline^\str_\tlit ~ p &::= s \mid e \\
    \expr^\str_\tlit ~ e &::= e^\str_\prog \mid \strtpl{t}
\end{align*}
For example, the expression \mintinline{Js}|`${x}b${x}`| would parse into the abstract syntax $\strtpl{\dtext{[x, \quot{b}, x]}}$.

Templates are fundamentally about providing a concise representation of document programs, as opposed to increasing the expressiveness of the document language (in the sense used by \citet{felleisen1991expressive}). Therefore, we do not provide an operational semantics directly for templates, but rather provide a translation from templates to the underlying calculus. (We provide a static semantics in \Cref{sec:flat-lists}). More precisely, the translation is specified as a family of functions over syntax kinds $\alpha$ with the form $\sug{\cdot}{\alpha} : \alpha \rightarrow \expr$.

Here, we arrive at a critical design question: \emph{how should templates translate to terms?} As described in the Scribble case study in \Cref{sec:scribble-study}, the particular choice of desugaring will influence how well templates compose with other language features. For example, a ``direct'' desugaring for string template literals might look like this:
\begin{align*}
\sug{[p_1, \ldots, p_n]}{\kttext} \,&\overset{?}{=}\, \sug{p_1}{\ktinline} + \ldots + \sug{p_n}{\ktinline} \\
\sug{\strtpl{t}}{\expr} \,&\overset{?}{=}\, \sug{t}{\kttext}
\end{align*}
In this desugaring, a template desugars to a term of type $\str$ produced by the concatenation of desugared template parts, and the $\msfb{strtpl}$ operator is just the identity. This desugaring is, in fact, a perfectly valid semantics for languages only supporting string template literals. For instance, the ECMAScript 2024 specification\,\cite[\S 13.2.8.6]{es2024} describes a roughly comparable evaluation strategy for ECMAScript template literals.

However, this direct desugaring does not easily generalize to higher levels of the document calculus. For instance, if a template part can be a variable binding $\tset{x}{e}$, it is not obvious how to desugar the binding under this general style of semantics. Or if we want to repurpose templates to generate trees of strings rather than just strings, then we want the desugaring to not collapse template parts into a single string too early.

Therefore, the $\dstrtmplstr$ semantics are carefully designed to support later levels. That semantics is given by the following desugaring:
\begin{align*}
    \sug{[]}{\kttext} &= [] \\
    \sug{p :: ps}{\kttext} &= \sug{p}{\ktinline} :: \sug{ps}{\kttext} \\
    \sug{s}{\ktinline} &= s \\
    \sug{e}{\ktinline} &= \sug{e}{\expr} \\
    \sug{\strtpl{\mtext}}{\expr} &= \msf{join} ~ \sug{\mtext}{\kttext}
\end{align*}

This desugaring represents a few key design decisions \textit{vis-\`a-vis} the direct desugaring. First, templates desugar to lists rather than strings. Second, the $\msfb{strtpl}$ operator is now responsible for converting the list to a string by desugaring to a $\msf{join}$. And third, the desugaring of a template is defined inductively, providing the opportunity for later desugarings of template parts to access the tail of the template.

Another possible desugaring for string template literals could follow the example of PHP by desugaring to effectful commands on a global mutable buffer. However, a pure semantics for templates avoids the compositionality issue described in  \Cref{sec:php-case-study} because the output of a template can be, for instance, combined with auxiliary data in a single data structure. Therefore we consider the effectful desugaring an anti-pattern and focus only on pure desugarings.

\subsubsection{String Template Programs}
\label{sec:string-tmpl-prog}

The final level of the document calculus in the string domain is the string template program calculus $\dstrtmpllang$. String template literals reduce the notation required to concatenate strings and expressions. However, complex templates often involve interpolating expressions which contain nested templates, requiring additional content delimiters. For example, compare the string template literal in Javascript (left) with the string template program in Jinja (right): 

\begin{twocol}
    \begin{col}
\begin{minted}{js}
// Javascript
var l = [1, 2, 3]
```
Examples of addition include: 
${
   l.map(n => `* ${n} + 1 = ${n + 1}`)
    .join("\n")
}
```
\end{minted}        
    \end{col}
    \begin{col}
\begin{minted}{jinja}
{# Jinja #}
{% set l = [1, 2, 3] -%}
Examples of addition include:
{% for n in l -%}
* {{ n }} + 1 = {{ n + 1 }}
{% endfor %}
\end{minted}        
    \end{col}
\end{twocol}

String template programs (more often called ``template languages'') offer concision by lifting computations such as binding and looping into the template. Intuitively, the difference is that in a normal program, content is delimited from computation, such as with quotes or backticks. In a template program, computation is delimited from content, such as with \verb|{% percents %}| in Jinja.

Formally, $\dstrtmpllang$ models this concept by adding support for if-expressions, set-statements, and foreach-loops:
\begin{align*}
  \ktinline^\str_\tprog ~ p ::= p^\str_\tlit \mid \tset{x}{e} \mid \tif{e}{t_1}{t_2}  \mid \dforeach{e}{x}{t}
\end{align*}

Set-statements must be desugared in the context of the rest of the template, so their desugaring is defined as special case over the syntactic kind $\kttext$ rather than $\ktinline$:
\begin{align*}
  \sug{(\tset{x}{e}) :: ps}{\kttext} &= \letexp{x}{e}{\sug{ps}{\kttext}}
\end{align*}
Observe here the importance of defining the desugaring of a template \emph{inductively} so as to permit such special cases, as opposed to independently desugaring each template part.

The semantics of $\msfb{if}$ and $\msfb{foreach}$ are definable at the $\dstrtmpllang$ level; however, we will delay introducing them until reaching the article domain (\Cref{sec:article-tmpl-prog}). These template parts contain nested templates and therefore desugar to nested lists, which requires a flattening/splicing mechanism to un-nest. The explanation of these mechanisms will be more enlightening when contrasted against article template programs as well as string template programs.

\subsection{The Article Calculus}
\label{sec:article-calc}

String template programs are the highest level of the document calculus in the domain of strings. Therefore, we can now proceed by enriching the domain with additional structure. The most common form of structured document is an \emph{attributed tagged tree} like this:
\begin{align*}
\knode ~ n &::= \textnode{s} \mid \node{s}{[(s, s)^*]}{n^*}
\end{align*}

For example, the introduction to this paper would be modeled like this:
\begin{align*}
    &\msfb{node} ~ (\quot{section}, [(\quot{id}, \quot{intro})], [ \\    
    & \ind \msfb{node} ~ (\quot{header}, [], [\textnode{\quot{Introduction}}]), \\
    & \ind \msfb{node} ~ (\quot{para}, [], [\textnode{\quot{We live in a golden age of document languages.}}, \ldots] \\
    &])
\end{align*}

Such trees can naturally represent many kinds of documents (e.g., HTML websites, XML data structures). We focus on the subset of tagged trees that represent \emph{articles}: a tree that consists of \emph{block} nodes (e.g., sections, paragraphs) and \emph{inline} nodes (e.g., plain text, bold text). More precisely, we define an article as an attributed tagged tree that adheres to this schema:
\begin{align*}
    \karticle~ a &::= b^* \\
    \kblock~ b &::= \node{\quot{para}}{[]}{\kappa} \mid \node{\quot{section}}{[]}{a} \\
    \ktext~ \kappa &::= \ell^* \\
    \kinline~ \ell &::= \textnode{s} \mid \node{\quot{bold}}{[]}{\kappa}
\end{align*}

\noindent For simplicity, this definition provides a minimal set of elements that are sufficient to model interesting aspects of article languages, to the exclusion of some common elements like italicized text and bulleted lists.  We will introduce additional elements and attributes as needed, such as when discussing references in \Cref{sec:references}.

We will present a sequence of document calculi $\dlm{\karticle}{\bullet}$ in the article domain, examining how the mechanisms previously examined for computing with strings can be lifted onto trees.

\subsubsection{Article Literals}
\label{sec:article-lit}

The lowest level in the article domain is the article literal calculus $\dartid$, analogous to the string literal calculus. This level models document languages like Markdown (left) and HTML (right):

\begin{twocol}
  \begin{col}
  \begin{minted}{markdown}
<!-- Markdown -->

Hello [world]!

[world]: https://example.com
  \end{minted}
  \end{col}
  \begin{col}
  \begin{minted}{html}
<!-- HTML -->
<p>
  Hello 
  <a href="https://example.com">world</a>!
</p>
  \end{minted}
  \end{col}
\end{twocol}
  
Formally, $\dartid$ just consists of article literals $a$:
$$\expr^\karticle_\lit ::= a$$ 

Note that even languages like Markdown are not fully literal. As in the example above, the feature of named URLs requires interpretation to resolve named references to URL definitions. We will discuss how to model this particular aspect further in \Cref{sec:references}.

\subsubsection{Article Programs}
\label{sec:article-prog}

The next level is the article program calculus $\dartprog$ of programs that construct articles through libraries in the base language. Article APIs usually look either like imperative widget trees as in Javascript (left), or functional combinators over lists as in Elm (right):

\begin{twocol}
  \begin{col}
    \begin{minted}{javascript}
// Javascript
let p = document.createElement("p");
let hello = 
  document.createTextNode("Hello ");
let world = 
  document.createElement("strong");
world.textContent = "world";
p.appendChild(hello);
p.appendChild(world);
    \end{minted}
  \end{col}
  \begin{col}
  \begin{minted}{elm}
-- Elm  
import Html exposing (text, p, strong)
main = p [] 
  [ text "Hello "
  , strong [] [ text "world"] 
  ]
\end{minted}
  \end{col}
\end{twocol}

\noindent Formally, $\dartprog$ is the combination of $\dstrprog$ and $\dartid$:
\begin{align*}
    \expr^\karticle_\prog ~ e &::= e^\str_\prog \mid e^\karticle_\lit
\end{align*}

\newcommand{\tynode}{\msf{NodeTy}}
\newcommand{\tystructnode}{\msf{struct\text{-}node}}

\noindent We model articles in $\dartprog$ as a sum of text nodes and structure nodes:
\begin{align*}
    \tau~\tystructnode &\triangleq \typrod{\msf{name}: \str,~  \msf{attrs}: \tylist{(\str \times \str)},~  \msf{children}: \tau} \\
    \tynode &\triangleq \mu \alpha.~ \tysum{\msf{text}: \str \mid \msf{node}: \tylist{\alpha}~\tystructnode} 
\end{align*}

Note that the type system of $\dartprog$ is expressive enough to evaluate if $\tc{\cdot}{e} {\tylist{\tynode}}$ for some document program $e$. However, the type system is not expressive enough to determine that $e$ is actually an article, e.g., that there are no block nodes nested in an inline node. This limitation is present in all existing document languages, except article literal languages like Markdown where only the article subset of trees is syntactically expressible.

\subsubsection{Article Template Literals}
\label{sec:article-tmpl-lit}

The next level is the article template literal calculus $\darttmpllit$. Article template literals are analogous to string template literals --- they are a pithy form of document constructor with support for expression interpolation, but they evaluate to articles rather than strings. The most common form of article template literal is either HTML syntax as in JSX Javascript (left), or XML syntax as in Scala 2 (also left) and Visual Basic .NET (right):

\begin{twocol}
  \begin{col}
    \begin{minted}{jsxlexer.py:JsxLexer -x}
// JSX Javascript and Scala 2
var items = 
  Array("Milk", "Eggs", "Cheese");
<article>
  <p>Today I am going shopping for:</p>
  <ul>
    {items.map(item => 
      <li><p>{item}</p></li>)}
  </ul>
</article>
    \end{minted}
  \end{col}
  \begin{col}
\newcommand{\mtag}[1]{\textcolor{mintedgreen}{\textbf{#1}}}
\newcommand{\mdelim}[1]{\textcolor{mintedyellow}{#1}}
\begin{minted}[escapeinside=||]{vbnet}
' VB.NET
Dim items() = 
  {"Milk", "Eggs", "Cheese"}
|<\mtag{article}>
  <\mtag{p}>Today I am going shopping for:</\mtag{p}>
  <\mtag{ul}>
    \mdelim{<%=} \mtag{From} item \mtag{in} items 
        \mtag{Select} <\mtag{li}>
          <\mtag{p}>\mdelim{<%=} item \mdelim{%>}</\mtag{p}>
        </\mtag{li}> \mdelim{%>}
  </\mtag{ul}>
</\mtag{article}>|
\end{minted}
  \end{col}
\end{twocol}

\noindent Observe that these syntaxes represent templates in part because all text inside the tags is undelimited, in contrast to the examples in \Cref{sec:article-prog}.

\newcommand{\treetplsup}{\msfb{\tiny treetpl}}

\newcommand{\mktext}[1]{\recsum{\msf{text}}{\tynode}{#1}}
\newcommand{\mknode}[1]{\recsum{\msf{node}}{\tynode}{#1}}

Formally, the semantics of $\darttmpllit$ are quite simple given our careful setup from earlier. In \Cref{sec:string-tmpl-lit} we defined $\kttext$ and $\ktinline$ to model string templates in $\dstrtmplstr$. We will reuse those features, adding a new template part for tree nodes, and adding a new template expression for articles:
\begin{align*}
    \expr^\karticle_\tlit~ e &::= e^\karticle_\prog \mid e^\str_\tlit \mid \treetpl{t} \\
    \ktinline^\karticle_\tlit~ p &::= p^\str_\tlit \mid \node{s}{[(s, e)^*]}{t}
\end{align*}
\begin{align*}
    \sug{\treetpl{t}}{\expr} &= \sugwith{t}{\kttext}{\tynode} \\
    \sugwith{s}{\ktinline}{\tynode} &= \mktext{s} \\
    \sugwith{\node{s}{at}{t}}{\ktinline}{\tynode} &= \mknode{\typrod{\msf{name}: s,~ \msf{attrs}: [(s, \sug{e}{\expr})^*],~ \msf{children}: \sug{t}{\kttext}}}
\end{align*}
 
\noindent For example, the following expression evaluates to the article containing a single paragraph with the text ``Hello \textbf{World}'':
$$
\letexp{x}{\quot{World}}{\treetpl{[\node{\quot{para}}{[]}{[\quot{Hello}, \node{\quot{bold}}{[]}{[x]}]}}}
$$

The key idea is that $\strtpl{t}$ should desugar to an expression of type $\str$, which in turn relies on $\sugwith{t}{\kttext}{\str}$ to desugar to an expression of type $\tylist{\str}$. For tree templates, we want $\treetpl{t}$ to desugar to an expression of type $\tylist{\tynode}$, which in turn relies on $\sugwith{t}{\kttext}{\tynode}$ to desugar to an expression of type $\tylist{\tynode}$. Therefore, there are two key differences in the desugarings of $\msfb{strtpl}$ and $\msfb{treetpl}$:
\begin{itemize}
    \item The desugaring of $\msfb{strtpl}$ wraps the template in a $\msf{join}$, while the desugaring of $\msfb{treetpl}$ does not.
    \item The desugaring of $\msfb{strtpl}$ has string literal template parts desugar to terms of type $\str$, while those same parts desugar to terms of type $\tynode$ inside a $\msfb{treetpl}$.
\end{itemize}

To implement the latter detail, we modify the desugaring function to be context-aware, notated with the superscript $\sugwith{\cdot}{}{\tau}$, where a template should desugar to a list containing elements of type $\tau$. The $\msfb{treetpl}$ desugaring enters the $\tynode$ context, which is assumed to be carried through where not explicitly written out. The $\msfb{strtpl}$ desugaring similar enters the $\str$ context, where string literals are desugared according to the rule:
$$
\sugwith{s}{\ktinline}{\str} = s
$$

\subsubsection{Article Template Programs}
\label{sec:article-tmpl-prog}

The final level is the article template program calculus $\darttmplprog$ of article templates with if-expressions, set-statements and foreach-loops. Examples include Typst (left) and Svelte Javascript (right):

\begin{twocol}
  \begin{col}
\begin{lstlisting}[language=typst,escapeinside=||]
// Typst
#let items = ("Milk", "Eggs", "Cheese")

Today I am going shopping |for|:

#for item in items [
  - #item
]
\end{lstlisting}
    \end{col}
\begin{col}
\begin{minted}[escapeinside=||]{jsxlexer.py:JsxLexer -x}
// Svelte Javascript
<script>
  let items = ["Milk", "Eggs", "Cheese"];
</script>
<article>
  <p>Today I am going shopping for:</p>
  <ul>
    |\{\#each items as item\}|
      <li><p>{item}</p></li>
    |\{/each\}|
  </ul>
</article>
\end{minted}
  \end{col}
\end{twocol}

Formally, $\darttmplprog$ requires no additional features, and can be generated by composing the last level of the string calculus with the previous level of the article calculus:
\begin{align*}
    \expr^\karticle_\tprog~ e ::= e^\karticle_\tlit \mid e^\str_\tprog
\end{align*}

For example, the shopping list program can be expressed as the following template (imagining the article domain is enriched with bulleted lists and list items):
\begin{align*}
&\msfb{treetpl}~ [ \\
&\ind \tset{\msf{items}}{[\quot{Milk}, \quot{Eggs}, \quot{Cheese}]}, \\
&\ind \node{\quot{para}}{[]}{[\textnode{\quot{Today I am going shopping for}]}}, \\
&\ind \node{\quot{list}}{[]}{
    [\dforeach{\msf{items}}{\msf{item}}{
        [\node{\quot{item}}{[]}{[\msf{item}]}]}]} \\
&]
\end{align*}

Now we have sufficient points of reference to return to the question first raised in \Cref{sec:string-tmpl-prog}: what are the semantics of $\msfb{foreach}$ and $\msfb{if}$? The crux of the problem is that these constructs contain nested templates, which affects the \emph{dimensionality} of the term desugared from a template. To explain, consider a simple desugaring of $\msfb{foreach}$ into a $\msf{map}$:
\begin{align*}
  \sug{\dforeach{e}{x}{t}}{\ktinline} &\overset{?}{=} \msf{map} ~ (\lambda x. ~ \sug{t}{\kttext}) ~ \sug{e}{\expr}
\end{align*}

\noindent Then consider the behavior of the example tree template under the simple desugaring. In particular, observe this part:
\begin{align*}
&\sug{\node{\quot{list}}{[]}{[\dforeach{\msf{items}}{\msf{item}}{
        [\node{\quot{item}}{[]}{[\msf{item}]}]}]}}{\ktinline} \\
        =~ &\node{\quot{list}}{[]}{[\msf{map} ~ (\lambda \msf{item}. ~ [\node{\quot{item}}{[]}{[\msf{item}]}]) ~ \msf{items}]} \\
        \bigsteptoop~ &\node{\quot{list}}{[]}{[[[\node{\quot{item}}{[]}{[\quot{Milk}]}], \ldots]]}
\end{align*}
Note that the child list of the \quot{list} node is 3-dimensional! That certainly does not match the article schema provided at the top of \Cref{sec:article-calc}. Any document language with a $\msfb{foreach}$-loop must somehow flatten the node list to one dimension; the key design question is where in the pipeline this should happen. A few different approaches can be found in existing languages:

\paragraph{Avoid nested node lists with an imperative template semantics.} 
For example, Svelte will translate the shopping list program into 132 lines of Javascript. Part of that translation is a function that constructs the DOM in an imperative manner, like this:
    \begin{minted}{javascript}
insert(target, article, anchor);
append(article, p);
append(article, t1);
append(article, ul);

for (let i = 0; i < each_blocks.length; i += 1) {
    each_blocks[i].m(ul, null);
}
    \end{minted}

    While this translation avoids the issue of list dimensionality, the use of imperative template semantics can cause issues as described for PHP in \Cref{sec:php-case-study}. In the case of Svelte, one limitation is that templates cannot be nested inside expressions. For instance, this program is not valid Svelte:
    \begin{minted}{jsxlexer.py:JsxLexer -x}
<ul>
  {items.map(item => <li><p>{item}</p></li>)}
</ul>
\end{minted}

\paragraph{Avoid nested node lists with unquote-splicing.} Quasiquotes are a kind of template language where the unquote-splicing operator can be used to reduce list dimensionality. For example, the shopping list could be written in Clojure with the Hiccup library\,\citeurl{https://weavejester.github.io/hiccup/} like this (noting the \verb|~@| unquote-splice):
\begin{minted}{clojure}
(def items ["Eggs", "Milk", "Cheese"])
(def item-lis (map #(h/html [:li [:p %]]) items))
(eval `(h/html [:ul ~@item-lis]))
\end{minted}

\noindent In $\darttmplprog$, this strategy is modeled by introducing unquote-splicing via a $\msfb{splice}$ template part:   
$$
\ktinline^\karticle_\tprog ~ p ::=~ p^\karticle_\tlit \mid p^\str_\tprog \mid \splice{e}
$$
\begin{align*}
  \sug{(\splice{e}) :: ps}{\kttext} =~ \msf{append}~\sug{e}{\expr}~\sug{ps}{\kttext}
\end{align*}
Like a set-statement, a $\msfb{splice}$ desugars in context as an append of the spliced head expression to the tail of the desugared template. Then $\msfb{foreach}$ and $\msfb{if}$ can be desugared into splices:
\begin{align*}
  &\sugwith{(\dforeach{e}{x}{\mtext}) :: ps}{\kttext}{\tynode} \\
  &= \sug{(\splice{\msf{flatten} ~ (\msf{map} ~ (\lambda x.~ \sug{\mtext}{\kttext}) ~ \sug{e}{\expr}})) :: ps}{\kttext} \\
  &\sugwith{(\tif{e}{t_1}{t_2}) :: ps}{\kttext}{\tynode} \\ 
  &= \sug{(\splice{\tif{\sug{e}{\expr}}{\sug{t_1}{\kttext}}{\sug{t_2}{\kttext}}}) :: ps}{\kttext}
\end{align*}

Because the desugaring generates another template (without $\msfb{foreach}$ or $\msfb{if}$ but with $\msfb{splice}$), the desugaring function must be recursively invoked. We prove that this desugaring produces well-typed terms (with well-typed inputs) in \Cref{sec:flat-lists}.

\paragraph{Permit nested node lists as a document IR, and flatten the IR later.} 
For example, document systems like Scribble, Typst, React, and ScalaTags\,\citeurlwith{https://com-lihaoyi.github.io/scalatags/}{com-lihaoyi.github.io/scalatags} provide a document IR that permits arbitrary levels of nesting. After a document program is interpreted to a value in the IR, a visitor sweeps through each node and recursively flattens all node lists.  

In $\darttmplprog$, this strategy can be modeled by introducing the concept of a \emph{fragment} as an arbitrarily nested list of document content:
\newcommand{\tyfrag}[1]{{#1}~\msf{fragment}}
\newcommand{\kfrag}{\msf{FNode}}
\newcommand{\tynodefrag}{\msf{NodeFrag}}
\begin{align*}
    \tyfrag{\tau} &\triangleq \mu \alpha. ~ \tysum{\msf{base}: \tau \mid \msf{children}: \tylist{\alpha}} \\
    \kfrag &\triangleq \mu \alpha.~ \tysum{\msf{text}: \str \mid \msf{node}: \tyfrag{\alpha} ~ \tystructnode} \\
    \tynodefrag &\triangleq \tyfrag{\kfrag}
\end{align*}

\newcommand{\fragtpl}[1]{\msfb{fragtpl}~{#1}}
\newcommand{\elimfrags}{\msf{elim\text{-}frags}}

\noindent The desugaring for $\msfb{foreach}$ and $\msfb{if}$ can then follow the ``simple'' desugaring described earlier, with some additional constructors to build terms of the appropriate type:
$$
\expr^\karticle_\tprog e ::= \ldots \mid \fragtpl{t}
$$
\begin{align*}
     &\sug{\fragtpl{t}}{\expr} &&= \elimfrags ~ (\recsum{\msf{children}}{\tynodefrag}{\sugwith{t}{\kttext}{\tynodefrag}}) \\
     &\sugwith{s}{\ktinline}{\tynodefrag} &&= \recsum{\msf{base}}{\tynodefrag}{(\recsum{\msf{text}}{\kfrag}{s}}) \\
     &\sugwith{\node{s}{at}{t}}{\ktinline}{\tynodefrag} &&= \recsum{\msf{base}}{\tynodefrag}{(\recsum{\msf{node}}{\kfrag}{(s, at, \sug{t}{\kttext})})} \\
     &\sugwith{\dforeach{e}{x}{t}}{\ktinline}{\tynodefrag} &&= \recsum{\msf{children}}{\tynodefrag}{(\msf{map} ~ (\lambda x. ~ \recsum{\msf{children}}{\tynodefrag}{\sug{t}{\kttext}}) ~ e)} \\
     &\sugwith{\tif{e}{t_1}{t_2}}{\ktinline}{\tynodefrag} &&= \recsum{\msf{children}}{\tynodefrag}{(\tif{\sug{e}{\expr}}{\sug{t_1}{\kttext}}{\sug{t_2}{\kttext}})}
\end{align*}

\noindent It is worth noting that these semantics would be simpler in a dynamically-typed language, as a nested list could be expressed with the standard list type rather than a bespoke fragment type, and all the fragment constructors would be obviated.

Under this semantics, a template $t$ should desugar to a term of type $\tynodefrag$. To flatten the fragment, we introduce a function $\elimfrags : \tynodefrag \rightarrow \tylist{\tynode}$:
\begin{align*}
    &\elimfrags ~ (\recsum{\msf{base}}{\tynodefrag}{(\recsum{\msf{text}}{\kfrag}{s})}) &&= [\mktext{s}] \\
    &\elimfrags ~ (\recsum{\msf{base}}{\tynodefrag}{(\recsum{\msf{node}}{\kfrag}{(s, at, f)})}) &&= [\mknode{(s, at, \elimfrags ~ f)}] \\
    &\elimfrags ~ (\recsum{\msf{children}}{\tynodefrag}{l}) &&= \msf{flatten} ~ (\msf{map} ~ \elimfrags ~ l)
\end{align*}

Both the splicing and fragment strategies are reasonable ways to deal with the template dimensionality problem in a functional manner. It is ultimately a matter of taste as to which should be used in practice. The unquote-splicing strategy seems intuitively cleaner from the perspective of the language designer (no messy intermediate representation), although the fragment strategy seems nicer from the perspective of the document author (no worrying about getting just the right combination of splices).

Note as well that any of these strategies will avoid Scribble's dropping-elements issue described in \Cref{sec:scribble-study}. The key idea is that templates desugar to lists, and in System F a list of strings or nodes cannot be mistaken for a sequence of expressions (unlike in Racket with its ``implicit \verb|begin|''). An equivalent $\darttmplprog{}$ program to the one in \Cref{sec:scribble-study} would produce the expected output.

\section{Extending the Document Calculus}
\label{sec:article-extensions}

\Cref{sec:doc-calc} described the semantics of the document calculus, arguing that it models the features of popular document languages. The next two sections demonstrate how the document calculus can provide a foundation for describing higher-level document features and for reasoning about document programs. In this section, we extend the document calculus with three interesting document features: references, reforestation, and reactivity. Each feature requires a non-trivial change to the language's semantics --- references require staged computation (\Cref{sec:references}), reforestation requires a global analysis of document structure (\Cref{sec:reforesting}), and reactivity requires a complex runtime (\Cref{sec:reactivity}). We also provide an OCaml implementation of each feature in the supplemental materials (note that these implementations are shallowly embedded so as to avoid the verbosity of System F).

\subsection{References}
\label{sec:references}

A common feature in \dl{}s is to support identifiers on nodes that can be referred to elsewhere in a document. Specifically, we consider an extension to the article schema where sections can have string identifiers, and a ref element can refer to a section:
\newcommand{\attrnode}[3]{\msfb{node}~({#1}, {#2}, {#3})}
\begin{align*}
\kblock~b &::= \ldots \mid \attrnode{\quot{section}}{[(\quot{id}, s)]}{a} \\
\kinline~\ell &::= \ldots \mid \attrnode{\quot{ref}}{[(\quot{target}, s)]}{[]}
\end{align*}

The intended semantics are comparable to \TeX's, i.e., the displayed content of a section reference should be the number of the referenced section. This feature brings two challenges: checking for invalid references, and computing the content of a reference.

\subsubsection{Reference Validity}
\label{sec:validity}

As alluded to in \Cref{sec:article-prog}, there are multiple conceptions of validity when thinking about document programs. For example, one form of validity is well-typedness of the input: a document expression $e$ is valid if $\tc{\cdot}{e}{\tylist{\tynode}}$. Another form of validity is well-formedness of the output: a document expression $e$ is valid if $e \bigsteptoop v$ and $v \in \karticle$. In the wild, validity sometimes means parseability: the CommonMark specification\,\cite{cmark-spec} for Markdown states that ``any sequence of characters is a valid CommonMark document''. 

\newcommand{\kidctxt}{\msf{IdCtxt}}
\newcommand{\idctxt}{\Delta}

As the document domain is enriched with additional structure, well-formedness becomes an insufficient criterion for document validity. In the case of references, an article is not valid if it references an unknown identifier, analogous to a free variable. Therefore, we need to model validity via an auxiliary judgment that captures whether a document is valid beyond its syntactic structure.

Formally, we model reference validity first by constructing a identifier context:
$$\kidctxt \triangleq (\str \ast \msf{int}~\msf{list})~\msf{list} \ind\ind \Delta : \kidctxt$$
This context maps identifiers to section numbers. We construct $\idctxt$ via the function:
$$\sectionsatf_\alpha :  \msf{int~list} \rightarrow \alpha \rightarrow \kidctxt \ast \msf{int}~\msf{list}$$
The key case that deals with sections is as follows:
\begin{align*}
&\sectionsatf_\tynode ~ (k :: k^*) ~ (\mknode{(\quot{section}, [(\quot{id}, id), \ldots], children)})  = \\
&\ind \letexp{(\idctxt, \_)}{\sectionsatf_{\tylist{\tynode}} ~ (1 :: k :: ks) ~ children}{} \\
&\ind \letexp{\idctxt'}{(id, k :: ks) :: \idctxt}{} \\
&\ind (\idctxt', (k + 1) :: ks)
\end{align*}
In this case, given a current section numbering $k :: ks$, the section's children are analyzed with a fresh subsection counter placed on the stack. The identifier context is updated with the current section's ID, and the section number is incremented.

Let $\sections{n} = \sectionsat{n, [1]}.0$. Then we can define a validity judgment $\wf{\idctxt}{\cdot}$, where an article $a$ is valid if $\wf{\sections{a}}{a}$. Two representative inference rules are as follows:
\begin{mathpar}
\inferrule{\wf{\idctxt}{b_1} \\ \cdots \\ \wf{\idctxt}{b_n}}{\wf{\idctxt}{\article{\listlit{b_1, \ldots, b_n}}}}

\inferrule
    {s \in \msf{dom}(\idctxt)}
    {\wf{\idctxt}{\mknode{(\quot{ref}, [(\quot{target}, s)], [])}}}
\end{mathpar}

\noindent The full validity judgment is provided in the OCaml implementation.

Representing references in documents has a similar flavor to representing binders in deeply embedded languages\,\cite{licata2009binding,cave2012binding}, and could in theory be addressed with similar techniques. One important difference is that in documents, both identifiers and references can be placed anywhere in the document; referential structure is not strictly hierarchical as with lexically-scoped variables.

\subsubsection{Reference Content}

The validity judgment must notably be expressed in two stages --- one to collect a context of identifiers ($\sections{a}$), and one to check for validity in that context ($\wf{\idctxt}{a}$). Similarly, the content of a reference must be generated in two stages. In \LaTeX, for example, a reference to the next section in this document like \verb|\ref{sec:reforesting}| will be replaced by the text ``\ref{sec:reforesting}''. This operation is non-local, because the document language cannot know the section number of a forward reference at the point of reference. Most document languages accomplish this task with a second pass over the document, such as the \verb|.aux| file generated by LaTeX on a first-pass which is later rendered on a second-pass.

\newcommand{\renderrefsf}{\msf{render\text{-}refs}}
\newcommand{\renderrefs}[1]{\renderrefsf(#1)}
\newcommand{\replacerefsf}{\msf{replace\text{-}refs}}
\newcommand{\replacerefs}[1]{\replacerefsf(#1)}

We model the generation of reference content in the document calculus as follows. Say that $\tc{\cdot}{e}{\tylist{\tynode}}$ and $e \bigsteptoop v$. Then we compute the final document $v' = \renderrefs{v}$ where $\replacerefsf_\alpha : \kidctxt \rightarrow \alpha \rightarrow \alpha$ is a visitor over articles, with the key case as follows:
\begin{align*}
&\replacerefsf_\tynode ~ (\Delta : \kidctxt) ~ (\mknode{(\quot{ref}, [(\quot{target}, id)], []}, \idctxt) : \tynode) =\\ 
&\ind \mktext{(\msf{section\text{-}number\text{-}to\text{-}string} ~ \idctxt[id])} \semsep
&\renderrefsf ~ (d : \tynode) = \replacerefsf ~ (\sectionsf ~ d) ~ d
\end{align*}

More broadly, this extension reflects a key aspects of computing with documents: the validity and content of a document program's output can be a global property of program, which requires commensurate features for non-local computation.

\subsection{Reforestation}
\label{sec:reforesting}




Another instance of non-local computation in documents is \emph{reforestation}. 
In some article \dl{}s, the document structure expressed by the programmer is often quite different from the final generated document structure. For example, languages like Typst, Scribble, and Markdown do not require paragraphs to be explicitly wrapped in a tag like \verb|<p>|; rather, paragraphs are inferred based on line breaks. Another common operation is to permit the programmer to write sections linearly, and then to reconstruct the section hierarchy by grouping content between pairs of headers.

In a \dl{} with reforestation, the user writes a template program which is initially evaluated into a ``raw'' document tree that is not a syntactically-valid article, but where all expressions have been reduced to a value. Then a second pass ``reforests'' the raw document into a syntactically-valid article by analyzing the global document structure of the input. For example, the \verb|decode| function in Scribble\,\cite[p.~113]{flatt2009scribble} implements this functionality.

\newcommand{\flowtpl}[1]{\msfb{flowtpl}~{#1}}
\newcommand{\reforest}[1]{\msfb{reforest}~{#1}}
\newcommand{\reforestf}{\msf{reforest}}
\newcommand{\reforestfn}[2]{\reforestf ~ {#1} ~ {#2}}

To model reforestation in the article calculus, we add a $\msfb{flowtpl}$ primitive for reforested tree templates, which desugars into a tree template wrapped in a call to a $\reforestf$ function:
\begin{align*}
    \expr^\karticle_\tprog e &::= \ldots \mid \flowtpl{t} \semsep
    \sug{\flowtpl{t}}{\expr} &= \msf{reforest} ~ \sug{\treetpl{t}}{\expr} ~ []
\end{align*}

The key detail is the implementation of $\msf{reforest} : \tylist{\tynode} \rightarrow \tylist{\tynode} \rightarrow \tylist{\tynode}$. The specifics vary between languages, but a simple example that we can implement for $\darttmplprog$ will collect inline elements into paragraphs.
The function $\reforestf$ iterates through a list of nodes $n^*$ with an accumulator for the current paragraph $par$. It emits paragraphs upon encountering the end of a list, a double newline (as in Markdown), or a block node. For example, the document on the left would be reforested to the document on the right:

\begin{twocol}
    \begin{minipage}{0.4\textwidth}
\begin{align*}
[&\mktext{\quot{Hello}}, \\
&\mktext{\quot{World}}, \\
&\mktext{\quot{\textbackslash{}n\textbackslash{}n}}, \\
&\mknode{(\quot{figure},~ [],~ [\ldots])}, \\
&\mktext{\quot{Post-figure}}]
\end{align*}
    \end{minipage}
    \begin{col}
\begin{align*}    
[&\msf{node}_\tynode~ (\quot{para},~ [],~ \\ 
&\ind[\mktext{\quot{Hello}},~ \mktext{\quot{World}}]), \\
&\mknode{(\quot{figure}, [],~ [\ldots])}, \\
&\mknode{(\quot{para}, [],~ [\mktext{\quot{Post-figure}}])}]
\end{align*}
    \end{col}
\end{twocol}

\noindent The precise definition of $\reforestf$ is as follows:
\begin{align*}
    \reforestfn{[]}{par} &= [\mknode{(\quot{para}, [], \msf{rev} ~ par)}] \\
    \reforestfn{(\mktext{\quot{\textbackslash n\textbackslash n}}) :: ns}{par} &=    
        (\mknode{(\quot{para}, [], \msf{rev} ~ par)}) :: \reforestfn{ns}{[]} \\
    \reforestfn{((\mktext{s}) :: ns)}{par} &= \reforestfn{ns}{(\mktext{s} :: par)} \\
   \reforestfn{(\mknode{(s, at, ns_c)} :: ns)}{par} &= \begin{cases}
        \mknode{(\quot{para}, [], \msf{rev} ~ par)} :: \\ 
        \hspace{1em}\mknode{(s, at, \reforestfn{ns_c}{[]})} :: \\ 
        \hspace{1em}\reforestfn{ns}{[]} & \text{if}~ \msf{is\text{-}block} ~ s \\[0.3em]
        \reforestfn{ns}{(\mknode{(s, at, cs)} :: par)} & \text{otherwise}
    \end{cases}
\end{align*}

Reforestation again demonstrates how document computation requires analysis of the global structure of a document, such as by accumulating sequential elements into groups. 
The correctness condition for the $\reforestf$ function is that it must generate a valid document that adheres to the $\karticle$ schema. This condition generally assumes that the input $\tylist{\tynode}$ also adheres to some intermediate schema as a precondition. For instance, the implementation above assumes that the input is already valid and does not, say, contain block nodes within inline nodes. A more aggressive implementation could attempt to repair an invalid document by, say, reordering invalid node nests. But in practice, document repair is most often performed during parsing rather than a later stage, as in HTML and Markdown.

\subsection{Reactivity}
\label{sec:reactivity}

Modern documents, especially those in the browser, can be \emph{reactive} to signals such as a timer or user input. Such reactions include animations, interactive widgets, and explorable explanations. Many recently-developed \dl{}s focus on reactivity, as we discuss in \Cref{sec:rw-reactive}. Therefore, it would be valuable to model reactivity within the document calculus. This model enables us to reason about how reactivity interacts with features like references, as we will discuss in \Cref{sec:reactive-refs}.

We model reactivity by blending ideas from two popular UI frameworks. First, we adopt functional reactive programming for UIs as in Elm\,\cite{czaplicki2013elm}, i.e., purely functional state management via signals. FRP is an appropriate paradigm for the document subset of UIs, and its purely functional nature fits well into our purely functional calculus. Second, we adopt UI components as in React, i.e., encapsulating model and view into a single object. Most existing reactive document languages use components (except Elm), so we reflect that fact in the model.

\newcommand{\component}[3]{\{\msf{init}: {#1},~ \msf{update}: {#2},~ \msf{view}: {#3}\}}
\newcommand{\cconstruct}[2]{\msfb{component}~{#1}~{#2}}
\newcommand{\instance}[5]{\{\msf{id}: {#1},~ \msf{component}: {#2},~ \msf{props}: {#3},~ \msf{state}: {#4},~ \msf{children}: {#5}\}}
\newcommand{\cinstance}[1]{\msfb{instance} ~ {#1}}

\newcommand{\einit}{e_\msf{init}}
\newcommand{\estate}{e_\msf{state}}
\newcommand{\eupdate}{e_\msf{update}}
\newcommand{\eview}{e_\msf{view}}
\newcommand{\eprops}{e_\msf{props}}
\newcommand{\echildren}{e_\msf{children}}

\newcommand{\docinit}[1]{\msf{doc\text{-}init}({#1})}
\newcommand{\docstepf}{\msf{doc\text{-}step}}
\newcommand{\docstep}[1]{\docstepf({#1})}
\newcommand{\docviewf}{\msf{doc\text{-}view}}
\newcommand{\docview}[1]{\docviewf({#1})}
\newcommand{\docstar}[1]{\msf{doc}\text{-$\ast$}({#1})}
\newcommand{\cview}[3]{\msf{cview}({#1}, {#2}, {#3})}
\newcommand{\reconcile}[3]{\msf{reconcile} ~ {#1} ~ {#2} ~ {#3}}

\newcommand{\tycomp}[2]{(#1, #2)~\msf{component}}
\newcommand{\tyinst}{\msf{instance}}
\newcommand{\typrops}{\tau_\msf{props}}
\newcommand{\tystate}{\tau_\msf{state}}
\newcommand{\tyreact}{\msf{ReactNode}}
\newcommand{\instid}{\msf{InstId}}
\newcommand{\mkinst}[1]{\recsum{\msf{inst}}{\tyreact}{#1}}

At the core of the reactive model are the types of components, instances, and nodes:

\begin{twocol}
    \begin{col}
\begin{align*}
&\tycomp{\typrops}{\tystate} \triangleq \\ 
    &\ind\{\msf{init} : \typrops \rightarrow \tystate,\\
    &\ind\,\, \msf{update} : \msf{Signal} \times \tystate \rightarrow \tystate, \\
    &\ind\,\, \msf{view} : \tystate \rightarrow \tyreact\}
\end{align*}        
    \end{col}
    \begin{col}
\begin{align*}
    &\tyinst \triangleq \exists \typrops, \tystate. \\
      &\ind\{\msf{id}: \instid, \\ 
      &\ind\,\, \msf{com}: \tycomp{\typrops}{\tystate} \\
      &\ind\,\, \msf{props}: \typrops, \\
      &\ind\,\, \msf{state}: \tystate, \\
      &\ind\,\, \msf{node}: \tyreact \}
\end{align*}        
    \end{col}
\end{twocol}
$$
\tyreact \triangleq \mu \alpha.~ \tysum{\msf{text}: \str \mid \msf{node}: \tylist{\alpha}~\tystructnode \mid \msf{inst}: \tyinst} 
$$

\noindent A component is described by an $\msf{init}$ function that converts properties into an initial state. The $\msf{update}$ function handles a signal (which we just assume $\msf{Signal} = \str$ for simplicity) and returns an updated state. The $\msf{view}$ function returns a reactive node for a given state. 

An instance is a particular reification of a component. It internally maintains a local state, and a child view, along with a unique ID that is generated by a function $\msf{gen\text{-}id} : () \rightarrow \instid$. A component is instantiated with the following function:
\newcommand{\inst}[2]{\msf{instantiate} ~ {#1} ~ {#2}}
\begin{align*}
    &\msf{instantiate} \triangleq \Lambda \typrops, \tystate .~ \lambda(com : \tycomp{\typrops}{\tystate}) .~ \lambda (props : \typrops) .\\
    &\ind \letexp{state}{com.\msf{init} ~ \msf{props}}{} \\
    &\ind \letexp{node}{com.\msf{view} ~ state}{} \\
    &\ind \msfb{pack} ~ \{\msf{id}: \msf{gen\text{-}id}~(),~ com,~ props,~ state,~ node\} ~ \msfb{as} ~ \tyinst
\end{align*}

\noindent Finally, a reactive node $\tyreact$ is a $\tynode$ with an additional case for instances.
\newcommand{\reacttpl}[1]{\msfb{reacttpl} ~ {#1}}
To integrate instances into templates, we add a $\msfb{reacttpl}$ expression and a $\msfb{component}$ template part:
\begin{align*}
    \expr^\karticle_\tprog~ e &::= \ldots \mid \reacttpl{t} \\
    \msf{TPart}~ p &::= \ldots \mid \msfb{component} ~ e_1 ~ e_2 \semsep
    \sug{\reacttpl{t}}{\expr} &= \sugwith{t}{\kttext}{\tyreact} \\
    \sugwith{\msfb{component} ~ e_1 ~ e_2}{\ktinline}{\tyreact} &= \mkinst{(\inst{e_1}{e_2})}
\end{align*}

\noindent For example, the following document uses a counter component that appends to a string every time the component is clicked:
\begin{align*}
    &\msfb{let} ~ \msf{counter} : \tycomp{\str}{\str \times \str} = \{ \\%
        &\ind\msf{init}: \lambda p.~ (p, \quot{}),\\
        &\ind\msf{update}: \lambda e.~ \lambda (p, s).~ \msfb{if} ~ e = \quot{click} ~ \msfb{then} ~ (p, p + s) ~ \msfb{else} ~ (p, s),\\
        &\ind\msf{view}: \lambda (\_, s).~ \mktext{s}\\
    &\} ~ \msfb{in} \\
    &\reacttpl{[\quot{The number of clicks is},~ \cconstruct{\msf{counter}}{\quot{|}}]}
\end{align*}

\noindent To make the document reactive, we must provide it a runtime. The runtime consists of two functions:
\begin{itemize}
    \item $\docstepf : (\instid, \msf{Signal})~\msf{map} \rightarrow \tyreact \rightarrow \tyreact$ takes a reactive document and a set of signals for each instance, and updates the state of each component with the signal.
    \item $\docviewf : \tyreact \rightarrow \tynode$ replaces instance nodes with their children, creating the final article to display.
\end{itemize}

\noindent Starting with an initial reactive document program $d_0$, the runtime iteratively generates views and steps in the following pattern:
\begin{center}
\begin{tikzcd}
d_0 \arrow[d, "\docviewf"] \arrow[r, "\docstepf"] & 
d_1 \arrow[d, "\docviewf"] \arrow[r, "\docstepf"] & 
d_1 \arrow[d, "\docviewf"] \arrow[r, "\docstepf"] & 
\cdots
\\ a_0 & a_1 & a_2
\end{tikzcd}
\end{center}

\noindent The key case of the step function for instances is as follows:
\begin{align*}
    &\docstepf ~ (signals : (\instid, \msf{Signal})~\msf{map}) ~ (\mkinst{inst}: \tyreact) = \\
    & \ind \msfb{if} ~ inst.\msf{id} \not\in signals ~ \msfb{then} ~\mkinst{\{{inst ~ \msfb{with} ~ \msf{node} = \docstepf ~ signals ~  inst.\msf{node}}\}} ~ \msfb{else} \\
    & \ind \letexp{state'}{inst.\msf{com}.\msf{update} ~ (signals[inst.\msf{id}]) ~ inst.\msf{state}}{} \\
    & \ind \letexp{node'}{inst.\msf{com}.\msf{view} ~ state'}{} \\
    & \ind \letexp{node''}{\reconcile{signals}{node}{node'}}{} \\
    & \ind\mkinst{\{inst ~ \msfb{with} ~ state'; children''\}}
\end{align*}

When an instance receives a signal, then the $\msf{update}$ function generates the new state, and the $\msf{view}$ function generates the new view. However, simply returning the new view would erase all the state contained in child instances. Therefore, we must \emph{reconcile} the old and new views, expressed with the $\msf{reconcile}$ function. The key case is as follows:
\begin{align*}
    &\reconcile{signals}{(\mkinst{i})}{(\mkinst{i'})}  \\ 
    &\ind= \begin{cases}
       \docstepf ~ signals ~ (\mkinst{i)} & \text{if} ~ i.\msf{com} = i'.\msf{com} \wedge i.\msf{props} = i'.\msf{props} \\
       \mkinst{i'} & \text{otherwise}
    \end{cases}
\end{align*}

\noindent If an instance has the same component and properties as before\footnotemark, then reconciliation persists its state and recursively steps the instance. Otherwise, the new instance is returned.

\footnotetext{
    Note that because $i$ and $i'$ are existentially-typed, it is not inherently type-safe to compare their fields. The runtime must provision some way of first determining whether the two instances have the same type.
}

Finally, the view function eliminates all instances from the node tree, with the key case as follows:
$$
\docviewf ~ (\mkinst{i}) = \docviewf ~ i.\msf{node}
$$

This runtime system is sufficient to model an Elm/React-like reactive document language, including per-component state and reconciliation on state updates. We provide an example of formal reasoning about this system in \Cref{sec:reactive-refs}.

\section{Reasoning with the Document Calculus}
\label{sec:reasoning}

Finally, we demonstrate the value of the document calculus as a \emph{formal} foundation (in addition to being a conceptual foundation) by reasoning about the semantics of document programs. Specifically, we prove two theorems: first, we prove that the template desugaring always produces terms of the correct type (\Cref{sec:flat-lists}). Second, we show how to design a provably correct implementation strategy for efficiently composing references and reactivity (\Cref{sec:reactive-refs}).

\subsection{Templates Desugar to Well-Typed Terms}
\label{sec:flat-lists}

We would like to be able to say that our particular desugaring of templates is ``correct'' by some metric. For example, the $\msfb{foreach}$ desugaring in \Cref{sec:article-tmpl-prog} involves both a $\msfb{splice}$ and a $\msf{flatten}$ --- we should be unable to prove some correctness theorem if the desugaring omitted either construct. 

One such theorem is the statement that \emph{templates desugar to well-typed terms}. Specifically, a sugared expression $\treetpl{t}$ should desugar to an expression with type $\tylist{\tynode}$. A sugared expression $\strtpl{t}$ should desugar to an expression with type $\str$. Of course, desugared template terms are only well-typed if used properly. For instance, a program cannot interpolate an expression of the wrong type, or use a unbound variable. That is to say: well-typed inputs lead templates to desugar to well-typed terms.

To capture these ideas, we extend the type system to describe the types of templates. The typing judgments for templates are systematically constructed from their desugaring, roughly following the ``type resugaring'' method developed by \citet{pombrio2018macrotypes}. To deal with the fact that template desugaring is dependent on context, we add a new kind of fact to the typing context that indicates the current template context:
\newcommand{\tplctx}[1]{\msf{tpl} ~ {#1}}
$$
\Gamma ::= \ldots \mid \tplctx{\tau}
$$

\noindent Then the typing rules for templates are as follows:

\begin{mathpar}
    
    \ir{T-StrTpl}
        {\tc{\Gamma, \tplctx{\str}}{t}{\tylist{\str}}}
        {\tc{\Gamma}{\strtpl{t}}{\str}}
        {t-strtpl}

    \ir{T-TreeTpl}
        {\tc{\Gamma, \tplctx{\tynode}}{t}{\tylist{\tynode}}}{\tc{\Gamma}{\treetpl{t}}
        {\tylist{\tynode}}}
        {t-treetpl}

    \ir{T-NilTpl}
        {\ }
        {\tc{\Gamma, \tplctx{\tau}}{[]}{\tylist{\tau}}}
        {t-niltpl}

    \ir{T-ConsTpl}
        {\tc{\Gamma, \tplctx{\tau}}{p}{\tau}  \\\\ 
         \tc{\Gamma, \tplctx{\tau}}{ps}{\tylist{\tau}}}
        {\tc{\Gamma, \tplctx{\tau}}{(p :: ps)}{\tylist{\tau}}}
        {t-constpl}
    
    \ir{TP-Str}
        {\ }
        {\tc{\Gamma, \tplctx{\tau}}{s}{\tau}}
        {tp-str}

    \ir{TP-Node}
        {\tc{\Gamma, \tplctx{\tau}}{t}{\tylist{\tau}} \\\\
         \forall i.~ \tc{\Gamma}{e_i}{\str}}
        {\tc{\Gamma,\tplctx{\tau}}{\node{s}{[(s_i, e_i)*]}{t}}{\tau}}
        {tp-node}

    \ir{TP-Set}
        {\tc{\Gamma}{e}{\tau_e} \\\\ 
         \tc{\Gamma, \tplctx{\tau}, x : \tau_e}{ps}{\tylist{\tau}}}
        {\tc{\Gamma, \tplctx{\tau}}{(\tset{x}{e} :: ps)}{\tylist{\tau}}}
        {tp-set}

    \ir{TP-Splice}
        {\tc{\Gamma}{e}{\tylist{\tau}} \\\\
         \tc{\Gamma, \tplctx{\tau}}{ps}{\tylist{\tau}}}
        {\tc{\Gamma, \tplctx{\tau}}{(\splice{e} :: ps)}{\tylist{\tau}}}
        {tp-splice}

    \ir{TP-Foreach}
        {\tc{\Gamma}{e}{\tylist{\tau_e}} \\
         \tc{\Gamma, \tplctx{\tau}, x : \tau_e}{t}{\tylist{\tau}} \\\\
         \tc{\Gamma, \tplctx{\tau}}{ps}{\tylist{\tau}}}
        {\tc{\Gamma, \tplctx{\tau}}{(\dforeach{e}{x}{t} :: ps)}{\tylist{\tau}}}
        {tp-foreach}

    \ir{TP-If}
        {
         \tc{\Gamma}{e}{\msf{bool}} \\
         \tc{\Gamma, \tplctx{\tau}}{t_1}{\tylist{\tau}} \\\\ 
         \tc{\Gamma, \tplctx{\tau}}{t_2}{\tylist{\tau}} \\ 
         \tc{\Gamma, \tplctx{\tau}}{ps}{\tylist{\tau}}}
        {\tc{\Gamma, \tplctx{\tau}}{(\tif{e}{t_1}{t_2} :: ps)}{\tylist{\tau}}}
        {tp-if}
\end{mathpar}

In general, all templates desugar to terms of type $\tylist{\tau}$ for some element type $\tau$. The rules for each template part lay out the conditions under which the overall template is well-typed. For example, a $\msfb{foreach}$ is well-typed if the input $e$ is a list, if the nested template $t$ is well-typed under the binding $x$, and if the tail $ps$ is well-typed.
If the rules are formulated correctly, then the following theorem should hold:
\begin{theorem}[Desugaring preserves types]
    Let $e \in \expr^\karticle_\tprog$. If $~\tc{\Gamma}{e}{\tau}$ then $~\tc{\Gamma}{\sug{e}{\expr}}{\tau}$.
\end{theorem}

We give the full proof by induction over the derivation of $\tc{\Gamma}{e}{\tau}$ in \appref{sec:proofs}, but here we can provide the intuition for one case, again focusing on $\msfb{foreach}$. Recall the desugaring of $\msfb{foreach}$ in $\msfb{treetpl}$:
\begin{align*}
\sugwith{(\dforeach{e}{x}{\mtext}) :: ps}{\kttext}{\tynode} &= \sug{(\splice{\msf{flatten} ~ (\msf{map} ~ (\lambda x.~ \sug{\mtext}{\kttext}) ~ \sug{e}{\expr})}) :: ps}{\kttext} \\
&= \msf{append} ~ (\msf{flatten} ~ (\msf{map} ~ (\lambda x.~ \sug{\mtext}{\kttext}) ~ \sug{e}{\expr})) ~ \sug{ps}{\kttext}
\end{align*}

By the inductive hypothesis, we can assume that:
\begin{mathpar}
    \tc{\Gamma}{\sug{e}{\expr}}{\tylist{\tau_e}}
    
    \tc{\Gamma, x : \tau_e}{\sug{t}{\kttext}}{\tylist{\tau}}

    \tc{\Gamma}{\sug{ps}{\kttext}}{\tylist{\tau}}
\end{mathpar}

\noindent Then the type of the desugared term can be systematically derived in standard fashion. The $\msf{map}$ term has type $\tylist{\tylist{\tau}}$. The $\msf{flatten}$ term therefore has type $\tylist{\tau}$. The $\msf{append}$ term therefore has type $\tylist{\tau}$. We conclude that the sugared and desugared terms have the same type.

\subsection{Correctly Composing References and Reactivity}
\label{sec:reactive-refs}

As shown in the React case study in \Cref{sec:react-toc}, it requires careful thought to correctly compose efficient reactivity with document features like section references. The content of a reference is a global property of a document based on the number of sections and the location of each section label. This global computation is conceptually at odds with reactivity, which is oriented towards localizing computation to components that aren't aware of their sibling or parent components.

\newcommand{\postprocessf}{\msf{postprocess}}

A simple approach to composing these extensions is to postprocess every reactively-generated document. The simple reactive runtime looks like this:
\begin{center}
\begin{tikzcd}
v_0 \arrow[d, "\docviewf"] \arrow[r, "\docstepf"] & 
v_1 \arrow[d, "\docviewf"] \arrow[r, "\docstepf"] & 
v_1 \arrow[d, "\docviewf"] \arrow[r, "\docstepf"] & 
\cdots
\\ a_0 \arrow[d, "\renderrefsf"] & 
a_1 \arrow[d, "\renderrefsf"] & 
a_2 \arrow[d, "\renderrefsf"]
\\ a_0' & a_1' & a_2'
\end{tikzcd}
\end{center}

However, this strategy is needlessly inefficient. For example, the counter component described earlier in \Cref{sec:reactivity} would never affect the section ordering, and therefore never affect the content of a reference. The $\kidctxt$ $\idctxt$ could be computed once on $a_0$ and then reused for all subsequent computations, or at least until $\idctxt$ is invalidated. For example, if the context was persisted after the first step and invalidated on the second step, then such a strategy would look like this:
\begin{center}
\begin{tikzcd}[column sep=huge]
v_0 \arrow[d, "\docviewf"]  \arrow[r, "\docstepf"] & 
v_1 \arrow[d, "\docviewf"] \arrow[r, "\docstepf"] & 
v_1 \arrow[d, "\docviewf"] \arrow[r, "\docstepf"] & 
\cdots \\
a_0 \arrow[dd, bend right] \arrow[d, "\sectionsf"]  & 
a_1 \arrow[dd, "\replacerefsf" near end] & 
a_2 \arrow[dd, bend right] \arrow[d, "\sectionsf"]  \\ 
\idctxt_0 \arrow[dr, bend left] \arrow[d, "\replacerefsf"] & & \idctxt_1 \arrow[d, "\replacerefsf"] \arrow[dr, bend left]  \\
a_0' & a_1' & a_2' & \ldots
\end{tikzcd}
\end{center}

\newcommand{\dirtyf}{\msf{dirty}}
\newcommand{\dirty}[1]{\dirtyf(#1)}
\newcommand{\naive}[1]{\msf{simple}(#1)}
\newcommand{\incr}[1]{\msf{incr}(#1)}

These two strategies can be formalized as functions $\msf{simple}$ and $\msf{incr}$ that take a given article $a_i$ and produce a final article $a_i'$. Their semantics are as follows:
\begin{align*}
    \naive{a_i} &= \renderrefs{a_i} \\
    \idctxt_i^\msf{incr} &= \begin{cases}
       \sections{a_i} & \text{if} ~ i = 0 \vee \dirty{v_{i - 1}, v_i} \\
       \idctxt^\msf{incr}_{i - 1} & \text{otherwise}
    \end{cases} \\
    \incr{a_i} &= \replacerefs{a_i, \idctxt^\msf{incr}_i}
\end{align*}

The $\dirtyf$ function is the key logic that determines whether $\idctxt$ should be recomputed on a given step. Before considering a specific implementation of $\dirtyf$, we can first articulate a correctness condition for this optimization: the incremental strategy should produce an equivalent document as the naive strategy for all inputs. Formally:

\begin{theorem}[Correctness of incremental strategy for section references]
    Let $e \in \expr_\tprog^\karticle$ where $\tc{\cdot}{e}{\tyreact}$. Let $e \bigsteptoop v_0$. Let $i \in \mathbb{N}$ and $v_i = \docstepf^i(v_0)$. Let $a_i = \docview{v_i}$. Then $\naive{a_i} = \incr{a_i}$.
\end{theorem}

This theorem reduces to the lemma that $\dirtyf$ will always be true if the section IDs have changed from one document to the next, or formally:

\begin{lemma}[Dirty function always catches a change to section numbering]
\label{lem:dirty}
$$\sections{a_{i-1}} \neq \sections{a_{i}} \implies \dirty{v_{i-1}, v_i}$$
\end{lemma}

For example, a simple implementation of $\dirtyf$ can recognize that the document structure can only change as a result of components. In this implementation, $\dirtyf$ only returns true if any component's children includes a section either before or after the step. Say we have a function $\msf{descendents} : \tyreact \rightarrow \str ~ \msf{set}$ that returns the types of nodes descendent from the input. Then $\dirtyf$ is as follows:
\newcommand{\descendents}[1]{\msf{descendents} ~ {#1}}
\begin{align*}
    &\dirty{\mkinst{inst_{i-1}}, \mkinst{inst_i}} = \\
    &\ind \quot{section} \in (\descendents{inst_{i-1}.\msf{node}} \cup \descendents{inst_i.\msf{node}}) \\ 
    &\ind \vee \dirty{inst_{i-1}.\msf{node}, inst_i.\msf{node}}
\end{align*}

In \appref{sec:proofs} we give a proof of how the theorem reduces to the lemma, as well as the proof of the lemma for this particular definition of $\dirtyf$. More broadly, the point is that this theorem demonstrates how the document calculus provides a foundation for reasoning about aspects such as how global document dependencies compose with local reactive computations.

\section{Related Work}
\label{sec:relatedwork}

The impetus for this work is that, in fact, very little work has attempted to provide a formal foundation for document languages from a computational perspective. An informal description of the Scribe language\,\cite{reid1980scribe} was published in POPL 43 years ago. \textit{The \TeX{}book}\,\cite{knuth1986texbook} gives a fairly precise specification for \TeX, but its principal concerns are parsing and rendering---less so the computation in the middle. The @-syntax of Scribble\,\cite{flatt2009scribble} is well-defined but its metatheory is not, leading in part to issues as in \Cref{sec:scribble-study}. Within efforts to formalize languages with templates like PHP\,\cite{filaretti2014php}, templates are usually a small footnote within the broader project, rather than a central focus of investigation.

In the rest of this section, we focus on understanding the practical systems that we model in this paper. Document languages have come a long way since the 1960s when ``only a few dozen people in the world knew how to typeset mathematical formulas''\,\cite{knuth1996talk}. Most academic research on document languages in the 20th century focused on vocabulary and abstractions for the graphical aspects of documents (\Cref{sec:markup-systems}). Such work largely continues today in the form of the ever-growing complexity of web browser rendering engines. The 1990s saw an explosion of languages for generating strings with templates (\Cref{sec:template-systems}). In the new millennium, document language research has shifted focus to the aspects more at the heart of our paper, namely using computation to generate articles (\Cref{sec:computational-documents}). Today, the most complex interactions between content and computation can be found in reactive document systems (\Cref{sec:rw-reactive}).

\subsection{Markup Systems}
\label{sec:markup-systems}

``Markup languages'' have a long history as programming languages for marking up documents that are presented on paper (literally via printing, or metaphorically in a PDF), or presented in the browser.
\citet{coombs1987markup} developed an early markup theory that distinguished ``procedural markup'', or low-level graphical commands, from ``descriptive markup'', or high-level structuring of a document. Descriptive markup, later called a ``ordered hierarchy of content objects''\,\cite{derose1997text}, formed the basis of systems such as SGML\,\cite{goldfarb1990sgml} and HyTime\,\cite{newcomb1991hytime} that would go on to inspire HTML and XML. The article domain of the document calculus defined in \Cref{sec:article-calc} is a model for descriptive markup, the predominant model for document languages used today. (The notable exception to this is \TeX, which evaluates into procedural markup but uses ``environments'' to attempt to simulate the experience of  descriptive markup.)

Early markup systems had relatively primitive support for computation. The Scribe system\,\cite{reid1980scribe} only supported custom environments that were composed out of a fixed set of formatting attributes. 
This tradition has continued with modern markup languages. Languages like Markdown\,\cite{cmark-spec}, AsciiDoc\,\citeurl{https://asciidoc.org/}, reStructuredText\,\citeurl{https://docutils.sourceforge.io/rst.html}, Markdoc\,\citeurl{https://markdoc.dev/}, and Pandoc\,\citeurl{https://pandoc.org/} all have little native support for anything resembling computation. At most, these languages have capabilities for resolving references or performing textual substitution for global variables defined in an external config file.

One exception is \TeX{}, which has a powerful macro system. It is notable that only 1 out of 27 chapters of \textit{The \TeX{}book}\,\cite{knuth1986texbook} concerns macros, a reflection of how many tasks \TeX{} had to juggle at the time of its inception. Later computational markup systems like Typst\,\cite{maedje2022typst} reflect a significantly greater separation of concerns between computation and rendering.

\subsection{String Template Systems}
\label{sec:template-systems}

Unhygienic macro systems like the C preprocessor and the M4 processor\,\cite{m4} can be viewed as the ur-template-systems (what we call a \emph{string template program} in the document calculus). Format strings (i.e., \emph{string template literals}) also date back to the earliest programming languages such as the \verb|PICTURE| clause in COBOL, \verb|WRITE| command in Fortran, and \verb|printf| function in Algol 68. Variable interpolation in strings can be found in several early shell languages, and was later adopted by Perl and Tcl. If-statements and for-loops inside templates were popularized by PHP, which used these facilities primarily to generate strings of HTML. However, PHP is today the only widely-used general-purpose programming language (to our knowledge) with built-in support for string template programs --- such features in other languages are expressed usually through domain-specific languages such as Jinja\,\citeurl{https://jinja.palletsprojects.com/} for Python, Handlebars\,\citeurl{https://handlebarsjs.com/} for Javascript, or Liquid\,\citeurlwith{https://shopify.github.io/liquid/}{shopify.github.io/liquid} for Ruby. 

Using string template programs to generate article literals can be prone to error, as discussed in the PHP case study in \Cref{sec:php-case-study}. This observation motivated \citet{parr2004template}, who gave one of the first formal models of string template literals. Parr's model is roughly equivalent to the $\dstrtmplstr$ level of the document calculus. Parr's goal was to reason about the computability of template languages (so as to demonstrate that a restricted template language was not Turing-complete), while our goal is to provide a concise model for a wide variety of template features.

\subsection{Article Template Systems}
\label{sec:computational-documents}

As XML and HTML gained popularity, many languages and libraries were developed to provide better ways of creating and analyzing tree-shaped documents. XDuce\,\cite{hosoya2003xduce}, XML Schema\,\cite{simeon2003xml}, and $\mathbb{C}$Duce\,\cite{benzaken2003cduce} provided for typed processing of XML documents, with a focus on encoding domain-specific XML schemata into the type system. The JWIG extension\,\cite{christensen2003jwig} to Java supported XML templates with holes to address the problems of generating structured documents via strings. ``Syntax-safe'' string template engines were designed to ensure that generated strings matched a schema\,\cite{arnoldus2007repleo,heidenreich2009safetpl}.
Scala even supported XML template literals upon its release, although XML support has since been deprecated.

Metaprogramming systems share many similarities with article template systems; programs are trees, and articles are trees. As described in \Cref{sec:article-tmpl-prog},  HTML libraries in modern Lisps permit the use of quasiquotes to generate HTML. To reduce the verbosity of writing string content within Lisp-embedded articles, Skribe's Sk-expressions\,\cite{gallesio20015skribe} and Scribble's @-expressions\,\cite{flatt2009scribble} provide for concise article templates. Our goal of providing statically-typed templates also overlaps with typed metaprogramming systems like MetaML\,\cite{taha1997metaml}, Template Haskell\,\cite{sheard2002haskell}, and Scala macros\,\cite{burmako2013scala}, just to name a few notable systems among many.

A goal of some article template systems is to statically check for the validity of documents, i.e., that a node tree matches a given document schema. XDuce and $\mathbb{C}$Duce are specially designed for this purpose by supporting the langauge of regular trees as types. Haskell libraries like \verb|type-of-html|\,\citeurlwith{https://github.com/knupfer/type-of-html}{knupfer/type-of-html} show how specific schema (like HTML) can be encoded into a sufficiently expressive type system of a general-purpose language. Our goal is to integrate templates into System F in as simple a manner as possible, so the document calculus only checks for the weaker property of well-formedness.

\subsection{Reactive Article Template Systems}
\label{sec:rw-reactive}

While interest in XML templates has since waned, interest in HTML templates has grown dramatically, especially focusing on templates that describe reactive HTML documents. Recent years have seen many new languages for authoring reactive articles: Idyll\,\cite{conlen2018idyll}, MDX\,\citeurl{https://mdxjs.com/}, Observable\,\citeurl{https://observablehq.com/}, and Living Papers\,\citeurlwith{https://github.com/uwdata/living-papers/}{uwdata/living-papers}. In particular, the JSX extension to Javascript, originally created by the developers of React\,\citeurl{https://reactjs.com/}, is now widely adopted within the Javascript ecosystem. Frameworks like Vue\,\citeurl{https://vuejs.org/} and SolidJS\,\citeurl{https://solidjs.com/} use JSX, and Svelte\,\citeurl{https://svelte.dev/} uses a JSX-like syntax. 

Notably, these reactive JS frameworks all provide substantively different desugarings for JSX into vanilla Javascript. The desugaring provided in \Cref{sec:reactivity} is most similar to React's, where the desugaring is straightforward and the runtime does most of the work. However, more recent frameworks like Svelte have adopted much more complex desugarings to improve efficiency. Svelte statically analyzes its templates for data dependencies to determine when components should react to state changes, avoiding the cost of dynamic dependency analysis. This trend provides a fertile ground for future PL research that can build on the foundations of the document calculus. Just as one example, Svelte's dependency analysis is deeply unsound, as it is not sensitive to fields or aliases (see: \href{https://svelte.dev/tutorial/updating-arrays-and-objects}{svelte.dev}). Future document languages will need firm theoretical foundations to correctly analyze and desugar complex templates.

\section{Discussion}
\label{sec:discussion}

This paper has presented the document calculus, a formal model for how templates interleave content and computation to produce strings and articles. Our immediate goal with this work is to provide a formal model that can undergird any theoretical investigation into document languages. \Dl{}s have long been a subject with a plethora of practice but only tacit theory, especially with regards to the computation/content boundary. 

Our long-term goal with this work is to provide conceptual clarity to designers of document languages. We hope that the vocabulary and semantics of the document calculus can guide future designs (this paper came out of the authors' own work in designing a document language). We conclude by discussing actionable takeaways for \dl{} designers (\Cref{sec:takeaways}), and then discuss one of the major challenges unaddressed in this paper: concrete syntax (\Cref{sec:concrete-syntax}).


\subsection{Implications for Language Designers}
\label{sec:takeaways}

The taxonomy in \Cref{tab:taxonomy} provides a high-level vocabulary for talking about the design space of \dl{}. A \dl{} designer should ask: which domain and constructors are most appropriate for their context of use? With regards to constructors, template literals can be quite powerful with an expressive language of expressions. But an imperative language with limited expressions would probably need template programs or else the template DSL is too limiting. 

With regards to domain, a language designer should be aware that designing templates for both strings and articles does not require wholly different features. A shared template language can be used across both domains; the language just needs different syntaxes for invoking a template in each domain context (i.e., a string template $\msfb{strtpl}$ versus a tree template $\msfb{treetpl}$).

The semantics in Sections \ref{sec:string-calc} and \ref{sec:article-calc} provide one possible implementation strategy for template desugaring. We strongly recommend a pure strategy over an impure strategy for the reasons discussed in the PHP case study (\Cref{sec:php-case-study}). We recommend a variable-binding desugaring that permits lexical scope and does not require a single global context (\Cref{sec:string-tmpl-lit}). We also recommend carefully considering the dimensionality of lists produced by each template feature to avoid dimension mismatch, such as by providing a splice/quasiquote feature (\Cref{sec:article-tmpl-prog}).

Designers should be aware that reducing expressions to values is not likely to be the last step in the document generation pipeline (\Cref{sec:article-extensions}). Global passes such as section numbering (\Cref{sec:references}) and reforestation (\Cref{sec:reforesting}) should execute on the reduced ``raw'' document. These passes have a subtle interaction with reactivity (\Cref{sec:reactivity}). We provide an example for how these separate concerns can be composed in \Cref{sec:reactive-refs}.

\subsection{Concrete Syntax for Language Users}
\label{sec:concrete-syntax}

Concrete syntax is not a common concern in programming languages research, where it is assumed to be handled through standard parsing techniques. But concrete syntax is \emph{essential} to \dl{}s, moreso than most other kinds of programming languages (on par with languages for novices or compact DSLs). The fundamental utility of a \dl{} is predicated on its syntactic convenience. Most authors would not want to write at only the $\dartprog$ level, like this:
\vspace{-0.5em}
\begin{minted}{javascript}
[text("We live in a "), bold(text("golden age")), text(" of documents.")]
\end{minted}

\vspace{-0.5em}
In that sense, this paper's subtitle is deliberately inaccurate: lambda is not the ultimate document. As Olin Shivers wrote, ``lambda is not a universally sufficient value constructor,'' and that holds true for constructing documents as well. To that end, future work on document languages should investigate the design of syntaxes that trade-off intuitiveness, error-tolerance, and systematicity.

For instance, Markdown's syntax is designed to be reasonably intuitive and maximally error-tolerant, at the expense of systematicity. Taking one example from \citet{beyond-markdown}, Markdown does not have a consistent strategy for parsing lists adjacent to a paragraph. A different behavior occurs depending on whether the list number is equal to 1 or not, as shown in this Markdown program (left) with its HTML output according to CommonMark (right):

\begin{twocol}
    \begin{col}
\begin{minted}{markdown}
A paragraph
1. A list

Another paragraph
2. Another list    
\end{minted}        
    \end{col}
    \begin{col}
\begin{minted}{html}
<p>A paragraph</p>
<ol><li>A list</li></ol>
<p>Another paragraph
2. Another list</p>
\end{minted}
    \end{col}
\end{twocol}

More generally, the widespread adoption of Markdown has demonstrated the strong desire for a concise document syntax. Yet, authors also want computation to simplify authoring of complex documents, as evinced by both the enduring usage of \LaTeX{} and the proliferation of ``Markdown++'' successor languages. It is an open question how to get the best of both worlds --- a human-friendly, concise syntax with a principled, powerful semantics. Now is clearly the time to revisit the accumulated design decisions of past languages to build the foundations for a better-documented future.

\section*{Data Availability Statement}

The archival version of the artifact for this document is hosted on Zenodo\,\cite{crichton_2023_8409115}. The latest version of the artifact is hosted on GitHub\,\citeurlwith{https://github.com/cognitive-engineering-lab/document-calculus}{cognitive-engineering-lab/document-calculus}.

\begin{acks}
    This work was partially supported by a gift from Amazon and by the US NSF under Grant No.~2319014. We are grateful to the numerous individuals who have worked on the document languages that influenced our work.
\end{acks}

\bibliographystyle{ACM-Reference-Format}
\bibliography{bibs/misc}

\ifcameraready
\makeatletter
\par\bigskip\noindent\small\normalfont\@received\par
\makeatother
\fi

\newpage
\appendix
\section{Appendix}

\subsection{Definitions}
\label{sec:additional-rules}

Here, we provide any static and dynamic semantics omitted in \Cref{sec:doc-calc}.

\subsubsection{String literal calculus} \ \\

\noindent Syntax:
\begin{align*}
    \val_\lit^\str~ v ::= s
\end{align*}

\noindent Static semantics:
\begin{mathpar}
    \inferrule{\ }{\tc{\Gamma}{s}{\str}}
\end{mathpar}

\subsubsection{String program calculus} \ \\

\noindent Syntax:
\begin{align*}
    \msf{TyCtxt}^\str_\prog~ \Gamma ::=~ &\cdot \mid \Gamma, x : \tau \mid \Gamma, \alpha \\
    \val^\str_\prog~ v &::= v^\str_\lit \mid \lambda (x : \tau).~ e \mid \Lambda \alpha.~ e \mid \{(\ell: v_\ell)^*\} \mid \inject{v}{\ell}{\tau} \mid \msfb{fold}_\tau~ v \\
    \msf{EvalCtxt}~ E^\str_\prog ::=~ &[\cdot] \mid E ~ e \mid v ~ E \mid \letexp{x}{E}{e} \mid \{(\ell : E_\ell)^*\} \mid E.\ell \\
    \mid~ & \inject{E}{\ell}{\tau} \mid \msfb{case} ~ E ~ \{(\ell(x) \Rightarrow e_\ell)^*\}  \\    
    \mid~ &\msfb{fold}_\tau ~ E \mid \msfb{unfold}_\tau ~ E \mid \msfb{pack} ~ (E, \tau_1) ~ \msfb{as} ~ \tau_2 \mid \msfb{unpack} ~ (x, \alpha) = E ~ \msfb{in} ~ e
\end{align*}

\noindent Static semantics:
\begin{mathpar}
\inferrule
    {\tc{\Gamma}{e_1}{\str} \\ \tc{\Gamma}{e_2}{\str}}
    {\tc{\Gamma}{e_1 + e_2}{\str}}

\inferrule
    {\tc{\Gamma, x : \tau_1}{e}{\tau_2}}
    {\tc{\Gamma}{\lambda(x : \tau_1).~ e}{\tau_1 \rightarrow \tau_2}}

\inferrule
    {\tc{\Gamma}{e_1}{\tau_1 \rightarrow \tau_2} \\
     \tc{\Gamma}{e_2}{\tau_1}}
    {\tc{\Gamma}{e_1 ~ e_2}{\tau_2}}

\inferrule
    {\tc{\Gamma, x : \tau}{e}{\tau}}
    {\tc{\Gamma}{\msfb{fix}(x : \tau). ~ e}{\tau}}

\inferrule
    {\tc{\Gamma}{e_1}{\tau_1} \\
     \tc{\Gamma, x : \tau_1}{e_2}{\tau_2}}
    {\tc{\Gamma}{\letexp{x}{e_1}{e_2}}{\tau_2}}

\inferrule
    {\forall \ell.~ \tc{\Gamma}{e_l}{\tau_l}}
    {\tc{\Gamma}{\{(\ell: e_\ell)^*\}}{\{(\ell: \tau_\ell)^*\}}}

\inferrule
    {\tc{\Gamma}{e}{\{(\ell: \tau_\ell)^*\}} \\
     \ell \in \{\ell_1, \ldots, \ell_n\}}
    {\tc{\Gamma}{e.\ell}{\tau_\ell}}

\inferrule
    {\tc{\Gamma}{e}{\tau_\ell} \\
     \ell \in \{\ell_1, \ldots, \ell_n\}}
    {\tc{\Gamma}{\inject{e}{\ell}{\tau_\ell}}{\tysum{(\ell: \tau_\ell)^*}}}

\inferrule
    {\tc{\Gamma}{e}{\tysum{(\ell: \tau_\ell)^*}} \\
     \forall \ell.~ \tc{\Gamma, x : \tau_\ell}{e_\ell}{\tau}}
    {\tc{\Gamma}{\msfb{case} ~ e ~ \{(\ell(x) \Rightarrow e_\ell)^*\}}{\tau}}

\inferrule
    {\tau = \mu \alpha.~ \tau' \\
     \tc{\Gamma}{e}{\tau'[\alpha \rightarrow \tau]}}
    {\tc{\Gamma}{\msfb{fold}_{\tau} ~ e}{\tau}}

\inferrule
    {\tau = \mu \alpha.~ \tau' \\
     \tc{\Gamma}{e}{\tau}}
    {\tc{\Gamma}{\msfb{unfold}_{\tau} ~ e}{\tau'[\alpha \rightarrow \tau]}}  

\inferrule
    {\tc{\Gamma, \alpha}{e}{\tau}}
    {\tc{\Gamma}{\Lambda \alpha.~ e}{\forall \alpha.~ \tau}}

\inferrule
    {\tc{\Gamma}{e}{\forall \alpha.~ \tau}}
    {\tc{\Gamma}{e[\tau']}{\tau[\alpha \rightarrow \tau']}}

\inferrule
    {\tc{\Gamma}{e_1}{\tau_2[\alpha \rightarrow \tau_1]}}
    {\tc{\Gamma}{\msfb{pack} ~ (e_1, \tau_1) ~ \msfb{as} ~ \exists \alpha.~ \tau_2}{\exists \alpha.~ \tau_2}}
  
\inferrule
    {\tc{\Gamma}{e_1}{\exists \alpha.~ \tau_1} \\
     \tc{\Gamma, x : \tau_1, \alpha}{e_2}{\tau_2}}
    {\tc{\Gamma}{\msfb{unpack} ~ (x, \alpha) = e_1 ~ \msfb{in} ~ e_2}{\tau_2}}
\end{mathpar}

\noindent Dynamic semantics:
\begin{mathpar}
\inferrule
  {e \mapsto e'}
  {E[e] \mapsto E[e']}

\inferrule
  {s_3 = s_1 + s_2}
  {s_1 + s_2 \mapsto s_3}

\inferrule
  {\ }
  {(\lambda (x : \tau).~ e) ~ v \mapsto e[x \mapsto v]}

\inferrule
  {\ }
  {\msfb{fix}(x : \tau).~ e \mapsto e[x \mapsto \msfb{fix}(x : \tau).~ e]}

\inferrule
  {\ }
  {\letexp{x}{v}{e} \mapsto e[x \mapsto v]}

\inferrule
  {\ }
  {\{(\ell : v_\ell)^*\}.\ell \mapsto v_\ell}

\inferrule
  {\ }
  {\msfb{case} ~ (\inject{v}{\ell}{\tau}) ~ \{(\ell(x) \Rightarrow e_\ell)^*\} \mapsto e_\ell[x \rightarrow v]}

\inferrule
  {\ }
  {\msfb{unfold}_\tau ~ \msfb{fold}_\tau ~ v \mapsto v}

\inferrule
  {\ }
  {v[\tau] \mapsto v}

\inferrule
  {\ }
  {\msfb{unpack} ~ (x, \alpha) = (\msfb{pack} (v, \tau_1) ~ \msfb{as} ~ \tau_2) ~ \msfb{in} ~ e \mapsto e[x \mapsto v]}
\end{mathpar}

\subsection{Proofs}
\label{sec:proofs}

Here, we provide proofs for the theorems articulated in \Cref{sec:reasoning}.

\begin{theorem}[Desugaring preserves types]
  Let $e \in \expr^\karticle_\tprog$. If $~\tc{\Gamma}{e}{\tau}$ then $~\tc{\Gamma}{\sug{e}{\expr}}{\tau}$.
\end{theorem}

The theorem hinges on a key lemma that relates the typing judgments for templates to template desugaring:

\begin{lemma}[Template desugaring produces lists of expected type]
  Let $t \in \kttext^\karticle_\tprog$. If $~\tc{\Gamma, \tplctx{\tau}}{t}{\tylist{\tau}}$ then $~\tc{\Gamma}{\sugwith{t}{\kttext}{\tau}}{\tylist{\tau}}$.
\end{lemma}

We will prove this lemma, and the theorem shortly follows.

\begin{proof}
Proceed by induction over the derivation of $\tc{\Gamma, \tplctx{\tau}}{t}{\tylist{\tau}}$.
\begin{itemize}
  \item \Cref{tr:t-niltpl}: $t = []$ and $\sug{t}{\kttext} = []$. Clearly $\tc{\Gamma}{[]}{\tylist{\tau}}$.
  \item \Cref{tr:t-constpl}: $t = p :: ps$ and $\tc{\Gamma, \tplctx{\tau}}{p}{\tau}$ and $\tc{\Gamma, \tplctx{\tau}}{ps}{\tylist{\tau}}$ and $\sug{p :: ps}{\kttext} = \sug{p}{\ktinline} :: \sug{ps}{\kttext}$. By the IH, $\tc{\Gamma}{\sugwith{ps}{\kttext}{\tau}}{\tylist{\tau}}$, so just need to show that $\tc{\Gamma}{\sugwith{p}{\ktinline}{\tau}}{\tau}$. Case on the derivation $p : \tau$:
  \begin{itemize}
    \item \Cref{tr:tp-str}: $p = s$. Case on the possible types of $\tau$:
    \begin{itemize}
        \item $\tau = \str$. Then $\sugwith{s}{\ktinline}{\str} = s$. Therefore $\tc{\Gamma}{s}{\str}$.
        \item $\tau = \tynode$. Then $\sugwith{s}{\ktinline}{\tynode} = \mktext{s}$. Therefore $\tc{\Gamma}{\mktext{s}}{\tynode}$.
    \end{itemize}
    \item \Cref{tr:tp-node}: $p = \node{s}{[(s_i, e_i)^*]}{t}$ and $\tc{\Gamma}{e_i}{\str}$ and $\tc{\Gamma, \tplctx{\tau}}{t}{\tylist{\tau}}$. Only possible type for $\tau$ is $\tau = \tynode$. Then:
    $$\sug{\node{s}{at}{t}}{\ktinline} = \mknode{\typrod{\msf{name}: s,~ \msf{attrs}: [(s, \sug{e}{\expr})^*],~ \msf{children}: \sugwith{t}{\kttext}{\tynode}}}$$
    By the IH, $\tc{\Gamma}{\sugwith{t}{\kttext}{\tynode}}{\tylist{\tynode}}$. Therefore $\tc{\Gamma}{\sugwith{p}{\ktinline}{\tynode}}{\tynode}$.
  \end{itemize}
  
  \item \Cref{tr:tp-splice}: $t = \splice{e} :: ps$ and $\tc{\Gamma}{e}{\tylist{\tau}}$ and $\tc{\Gamma, \tplctx{\tau}}{ps}{\tylist{\tau}}$. Then:
  $$
  \sug{(\splice{e}) :: ps}{\kttext} = \msf{append}~\sug{e}{\expr}~\sug{ps}{\kttext}
  $$
  By the IH, $\tc{\Gamma}{\sugwith{ps}{\kttext}{\tau}}{\tylist{\tau}}$. By the definition of $\msf{append}$, then $\tc{\Gamma}{\sugwith{t}{\kttext}{\tau}}{\tylist{\tau}}$.

  \item \Cref{tr:tp-if}: $t = \tif{e}{t_1}{t_2} :: ps$ and $\tc{\Gamma}{e}{\msf{bool}}$ and $\tc{\Gamma, \tplctx{\tau}}{t_1}{\tylist{\tau}}$ and $\tc{\Gamma, \tplctx{\tau}}{t_2}{\tylist{\tau}}$. Then:
  \begin{align*}
   &\sugwith{(\tif{e}{t_1}{t_2}) :: ps}{\kttext}{\tau} \\
   &= \sug{(\splice{\tif{\sug{e}{\expr}}{\sug{t_1}{\kttext}}{\sug{t_2}{\kttext}}}) :: ps}{\kttext} \\
   &= \msf{append} ~ (\tif{\sug{e}{\expr}}{\sug{t_1}{\kttext}}{\sug{t_2}{\kttext}}) ~ \sug{ps}{\kttext}
  \end{align*}
  As in the previous case, it suffices to show that the first argument to $\msf{append}$ has type $\tylist{\tau}$. This follows from the IH which shows that $\tc{\Gamma}{\sugwith{t_1}{\kttext}{\tau}}{\tylist{\tau}}$ and similarly for $t_2$.

  \item \Cref{tr:tp-foreach}: $t = \dforeach{x}{e}{t} :: ps$ and $\tc{\Gamma}{e}{\tylist{\tau_e}}$ and $\tc{\Gamma, x : \tau_e, \tplctx{\tau}}{t}{\tylist{\tau}}$. Then:
  \begin{align*}
    &\sugwith{(\dforeach{e}{x}{\mtext}) :: ps}{\kttext}{\tynode} = \\ 
    &= \sug{(\splice{\msf{flatten} ~ \msf{map} ~ (\lambda x.~ \sug{\mtext}{\kttext}) ~ \sug{e}{\expr}}) :: ps}{\kttext} \\ 
    &= \msf{append} ~ (\msf{flatten} ~ (\msf{map} ~ (\lambda x.~ \sug{\mtext}{\kttext}) ~ \sug{e}{\expr})) ~ \sug{ps}{\kttext}
  \end{align*}
  Again, it suffices to show that the $(\msf{flatten} ~ \ldots)$ term has type $\tylist{\tau}$. By the definition of $\msf{map}$, because the lambda has type $\tau_e \rightarrow \tylist{\tau}$ and $\sug{e}{\expr} : \tylist{\tau_e}$, then the call has type $\tylist{\tylist{\tau}}$. The $\msf{flatten}$ call reduces the dimension to $\tylist{\tau}$.
\end{itemize}
\end{proof}

Now, we can prove the original theorem.
\begin{proof}
Proof by induction over the derivation of $\tc{\Gamma}{e}{\tau}$. For all terms, $\sug{e}{\expr} = e$ except for the case of templates (or expressions containing templates), so the theorem is trivial in those cases. We therefore focus on the two template cases:
\begin{itemize}
    \item \Cref{tr:t-strtpl}: $e = \strtpl{t}$ and $\tc{\Gamma, \tplctx{\str}}{t}{\tylist{\str}}$ and $\tc{\Gamma}{e}{\str}$. Then:
    $$
    \sug{\strtpl{t}}{\expr} = \msf{join} ~ \sugwith{t}{\kttext}{\str}
    $$
    By the lemma above, we have that $\tc{\Gamma}{\sugwith{t}{\kttext}{\str}}{\tylist{\str}}$. Then by the definition of $\msf{join}$, we have that $\tc{\Gamma}{\sug{e}{\expr}}{\str}$.

    \item \Cref{tr:t-treetpl}: $e = \treetpl{t}$ and $\tc{\Gamma, \tplctx{\tynode}}{t}{\tylist{\tynode}}$ and $\tc{\Gamma}{e}{\tylist{\tynode}}$. Then:
       $$
    \sug{\treetpl{t}}{\expr} = \sugwith{t}{\kttext}{\tynode}
    $$
     By the lemma above, we have that $\tc{\Gamma}{\sugwith{t}{\kttext}{\tynode}}{\tylist{\tynode}}$. Therefore, $\tc{\Gamma}{\sug{e}{\expr}}{\tylist{\tynode}}$.
\end{itemize}
\end{proof}

\begin{theorem}[Correctness of incremental strategy for section references]
    Let $e \in \expr_\tprog^\karticle$ where $\tc{\cdot}{e}{\tyreact}$. Let $e \bigsteptoop v_0$. Let $i \in \mathbb{N}$ and $v_i = \docstepf^i(v_0)$. Let $a_i = \docview{v_i}$. Then $\naive{a_i} = \incr{a_i}$.
\end{theorem}

\begin{proof}
    First, we substitute the definitions of relevant functions to obtain an elaborated correctness condition:
    \begin{align*}
    &\naive{a_i} = \incr{a_i} \\
    \iff& \renderrefs{a_i} = \replacerefs{a_i, \idctxt^\msf{incr}_i} \\
    \iff& \replacerefs{a_i, \sections{a_i}} = \replacerefs{a_i, \idctxt^\msf{incr}_i}
    \end{align*}
    From this condition, we can deduce:
    $$
    \sections{a_i} = \idctxt^\msf{incr}_i \implies \naive{a_i} = \incr{a_i}
    $$
    That is, to prove the full correctness property, it suffices to show that $\sections{a_i} = \idctxt^\msf{incr}_i$. We proceed by induction over $i$. 
    \begin{itemize}
        \item If $i = 0$, then $\sections{a_0} = \idctxt^\msf{incr}_0$ by definition of $\idctxt^\msf{incr}_0$.
        \item If $i > 0$, then assume $\sections{a_{i-1}} = \idctxt^\msf{incr}_{i-1}$.  Then we have:
        $$
        \dirty{v_{i_1}, v_i} \vee  \sections{a_{i-1}} = \sections{a_i} \implies \sections{a_i} = \idctxt^\msf{incr}_i
        $$
        The antecedent can be rewritten by logical equivalence:
        \begin{align*}
        &\dirty{v_{i-1}, v_i} \vee  \sections{a_{i-1}}  = \sections{a_i}\\
        \iff& \sections{a_{i-1}} \neq \sections{a_i} \implies \dirty{v_{i-1}, v_i}
        \end{align*}
    \end{itemize}
    Assume that $\sections{a_{i-1}} \neq \sections{a_i}$. Then we want to show that $\dirty{v_{i-1}, v_i}$. 

    By the definition of $\sectionsf$, there exists a node $n = \node{\quot{section}}{at}{ns}$ such that (WLOG) $n \not\in \descendents{v_{i-1}}$ and $n \in \descendents{v_i}$. By the definition of $\docstepf$, there exists an instance $inst$ in both $v_{i-1}$ and $v_i$ such that $n \in \descendents{inst}_i$. and $n \not\in \descendents{inst}_{i-1}$. By the definition of $\dirtyf$, then $\dirty{v_{i-1}, v_i}$ is true because  $n \in \descendents{inst}_i$.
\end{proof}

\setcounter{TotPages}{27}

\end{document}